\documentclass[letterpaper,12pt]{article}
\usepackage{amssymb}
\usepackage{amsmath}
\usepackage{amsfonts}
\usepackage{geometry,enumerate,bbm}
\usepackage{amsthm}
\usepackage{natbib}
\usepackage{booktabs}
\usepackage{epsfig}
\usepackage{float}
\usepackage[onehalfspacing]{setspace}
\usepackage{graphicx,color}
\usepackage{subcaption}
\usepackage{lscape}
\usepackage[flushleft]{threeparttable}
\usepackage[normalem]{ulem}
\usepackage{scrextend}
\usepackage{afterpage}
\usepackage{comment}
\definecolor{dgreen}{rgb}{0.0, 0.5, 0.13}

\usepackage{pgfplots}
\newcommand{\dd}{\mathrm{d}}

\RequirePackage[colorlinks,citecolor=blue,urlcolor=blue,linkcolor=blue]{hyperref}
\allowdisplaybreaks

\newtheorem{remark}{Remark} 
\newtheorem{assumption}{Assumption} 
\newtheorem{theorem}{Theorem} 
\newtheorem{example}{Example}

\newtheorem{definition}{Definition}
\newtheorem{lemma}{Lemma}
\newtheorem{corollary}{Corollary}

\newcommand{\diag}{{\rm diag}}

\DeclareMathOperator{\logit}{logit}

\begin{document}

\title{\bf Approximate Operator Inversion for Average Effects in Nonlinear Panel Models\thanks{
 We thank Stéphane Bonhomme for helpful comments. We are also grateful to seminar participants at the University of Oxford, Queen Mary University of London, Maastricht University, and Dongbei University of Finance \& Economics
for helpful comments.  Jad Beyhum and Geert Dhaene gratefully acknowledge financial support from the Research Fund of KU Leuven (grant STG/23/014) and from the Research Foundation – Flanders (grant G031125N).}}
\author{\setcounter{footnote}{2}
Jad Beyhum\thanks{%
Department of Economics, KU Leuven, \texttt{jad.beyhum@kuleuven.be} }
\and
Geert Dhaene\thanks{%
Department of Economics, KU Leuven, \texttt{geert.dhaene@kuleuven.be} }
\and
Cavit Pakel%
\thanks{%
Department of Economics, University of Oxford, \texttt{cavit.pakel@economics.ox.ac.uk} } 
\and
Martin Weidner%
\thanks{%
Dept.\ of Economics \& Nuffield College, University of Oxford, \texttt{martin.weidner@economics.ox.ac.uk} } }
\date{\today}

\maketitle
\thispagestyle{empty}
\setcounter{page}{1}
\begin{abstract}
    \noindent 
We study the estimation of average effects in nonlinear panel data models with fixed effects when the time dimension $T$ is only moderately large. Our approach, called approximate operator inversion (AOI), offers a new perspective on bias correction. Instead of first estimating unit-specific fixed effects and then correcting the resulting plug-in bias, AOI approximately inverts the likelihood-induced mapping from the fixed-effect distribution to the outcome distribution. AOI can be interpreted as the limit of an infinitely iterated bias correction scheme, and this limit is available in closed form. We show that the bias of the AOI estimator has a rate double robustness property and converges to zero at an exponential rate in $T$ under regularity conditions. Our asymptotic theory requires $T \to \infty$, but the exponential convergence rate of the bias means that finite-sample performance is very good even for moderately large $T$. We establish asymptotic normality and provide feasible inference.
\end{abstract}
\bigskip
\noindent
{\bf Keywords:} { Panel data, discrete choice, average effects, incidental parameters, ill-posed inverse problem}

\bigskip
\noindent
{\bf JEL classification code:} {C14, C23, C25}

\newpage 

\section{Introduction}
\label{sec:introduction}

Nonlinear panel data models with fixed effects are central to empirical work in economics.
In the standard semiparametric setting, 
the researcher specifies for each individual unit $i=1,\ldots,n$, a conditional probability (or density) $f(Y_i|X_i,A_i;\theta_0)$ of outcomes $Y_i = (Y_{i1}, \ldots, Y_{iT})$ given the observed covariates $X_i = (X_{i1}, \ldots, X_{iT})$, the fixed effects $A_i$ that capture unobserved heterogeneity, and the model parameter $\theta_0$. The conditional distribution of fixed effects $\pi_0(A_i|X_i)$, on the other hand, is left completely unspecified. Two objects of typical interest are the common model parameters $\theta_0$ and average 
effects of the form $\mu_0 = \mathbb{E}[\mu(X_i, A_i, \theta_0)]$, such as average 
marginal effects or average treatment effects. The presence of unobserved fixed effects leads to the well-known incidental parameter problem \citep{neyman1948consistent}, which complicates estimation of both $\theta_0$ and $\mu_0$. The literature frames the study of this problem along two asymptotic regimes, each with its own characteristic issues.

\begin{enumerate}[(A)]
    \item \textbf{Large $n$ and fixed $T$.} In this setting, each unit contributes only a finite number of observations. Consequently, the distribution $\pi_0(A_i|X_i)$ can only be set-identified. Essentially, the issue at hand is inversion (or, lack thereof) of the likelihood-induced mapping from the fixed-effect distribution to the outcome distribution, as given by
\begin{align}
    \mathbb{P}(Y=y\,| \, X=x) 
    = 
    \int_{\mathcal{A}} f(y\,| \, x,\alpha;\theta_0)\,
    \pi_0(\alpha\,| \, x)\,\dd\alpha.
    \label{eq:popnmain}
\end{align}
Lack of point identification of $\pi_0(A_i|X_i)$ precludes point identification of $\theta_0$ except under certain parametric families. More seriously, point-identification of $\mu_0$ largely fails. Except for some very specific cases, set identification is the norm rather than exception. 

 \item \textbf{Large $n$ and large $T$, with $A_i$ consistently estimable.} 
The analysis under this regime
reasons that as $T\to\infty$ each unit should accumulate enough information for $A_i$ to be consistently estimable. This abstracts away from identification of $\pi_0(A_i|X_i)$ and allows for inference based on $f(y\,| \, x,\alpha;\theta)$ using plug-in consistent estimates of $A_i$. However, even though consistently estimable, the asymptotically growing number of fixed effects leads to the incidental parameter bias. The bias-correction literature offers various options for removing this asymptotic bias. 
\end{enumerate}

Clearly, the large-$T$ literature operates on the level of $A_i$ and not $\pi_0(A_i|X_i)$, sidestepping the fundamental inversion problem in \eqref{eq:popnmain}. This poses an inherent limitation. Indeed, the large-$T$ literature almost exclusively focuses on the correction of the leading $O(1/T)$ bias. While this is entirely justified for truly large values of $T$, in practice the remaining bias terms can potentially matter for moderate values of $T$, especially in more complicated models. At a practical level, complete correction of bias remains a virtually hopeless task. At a theoretical level, the focus on $A_i$ rather than $\pi_0(A_i|X_i)$ limits the scope of traditional bias-correction methods.

Motivated by these observations, 
we propose the novel approximate operator inversion (AOI) method. This method can be used to conduct inference on both $\theta_0$ and $\mu_0$ under large-$n$ large-$T$ asymptotics, though our focus in this paper will be on $\mu_0$. The AOI estimator operates directly on the level of $\pi_0(A_i|X_i)$ rather than of $A_i$. Specifically, our approach
treats the estimation of $\mu_0$ as an operator inversion problem and it
approximately inverts the likelihood-induced mapping from the fixed-effect
distribution to the outcome distribution given in \eqref{eq:popnmain}.
By shifting the focus to the fundamental inversion problem, our approach bridges large-$T$ analysis with the fixed-$T$ literature. 

Crucially, although distinct from traditional bias-correction methods, the AOI estimator admits an iterated bias-correction interpretation and, in a well-defined sense, delivers infinite-order bias correction. This advances the state of the art in the large-$T$ panel literature, where existing methods typically correct bias only up to a fixed, often first, order.

The bias of AOI depends on two approximation components: how well the 
average-effect function $\mu(X, \cdot, \theta_0)$ is approximated by a chosen 
basis, and how well the unknown fixed-effect distribution is approximated by 
another basis.
The overall bias is governed by the product of these two approximation errors,
a property we call \textit{rate double robustness}. Rapid convergence in either component
is enough for fast overall bias decay. Under regularity conditions, the bias decays at an exponential rate in $T$, i.e., $O(\rho^T)$ for
some $0 < \rho < 1$. This \textit{exponential decay} property is a major improvement over the polynomial rates $O(1/T^s)$ obtained in 
standard bias-correction methods.

Operating on $A_i$ rather than $\pi_0(A_i|X_i)$ can result in more than the loss of infinite-order bias-correction. We highlight this by studying a third asymptotic scenario that interpolates between the large-$T$ and fixed-$T$ settings. 

\begin{enumerate}[(A)]
    \setcounter{enumi}{2}
    \item \textbf{Large $n$ and large $T$, with $A_i$ not consistently estimable.}
     Although $T\to\infty$, no unit is able to accumulate enough information on $A_i$: consistent estimation of $A_i$ fails even as $T$ tends to infinity. We show that $\theta_0$ and $\mu_0$ can nevertheless be consistently estimated
    in this regime under regularity conditions.
\end{enumerate}
The traditional large-$T$ asymptotics implicitly assumes that all fixed effects are consistently estimable as long as $T\to\infty$. However, and crucially, consistent estimation of $A_i$ is not a matter of whether $T$ tends to infinity but of whether the Fisher information in data on $A_i$ grows without bound as $T\to\infty$. Depending on the model and the covariate design, $A_i$ can fail to be consistently estimable, rendering existing bias-correction approaches non-operational. Our study of this scenario highlights that the extant bias-correction  literature operates not under large-$T$ asymptotics but a subset of it where $A_i$ is consistently estimable. We believe that our paper is the first one to make this distinction. Importantly, under regularity conditions, AOI is inherently immune to this issue.

The ideas in this paper apply to inference on both
$\theta_0$ and $\mu_0$. From Section~\ref{sec:Setup} onwards, however, we focus on
estimation of $\mu_0$, which presents the greater challenge. Even in models
where $\theta_0$ can be estimated consistently for fixed $T$ (for example,
using conditional likelihood methods), point identification of $\mu_0$
typically fails for fixed $T$ because the distribution of $A_i$ is only
partially identified from short panels. This motivates our focus on $\mu_0$
and the development of the AOI estimator for this purpose.

\paragraph{Related literature.}
AOI builds on two prior contributions.
The first is functional differencing \citep{bonhomme2012functional}, which 
constructs moment conditions that identify $\theta_0$ while remaining exactly 
free of $A_i$, without any large-$T$ approximation.
Such exact moment conditions exist only in special models and for specific
choices of $\mu_0$, and do not extend to average effects in general
\citep[see also][for applications to average marginal
effects]{dano2023transition, aguirregabiria2024identification}.
\cite{bonhomme2017panel} adopt a similar inverse-problem perspective on average-effect estimation, but restrict attention to settings where $\mu_0$ is point-identified at fixed $T$, ruling out the discrete choice models considered here. 

The second is approximate functional differencing \citep{dhaene2023approximate},
which replaces exact moment conditions with approximate ones whose error vanishes
as $T \to \infty$, broadening the scope to settings where exact moment conditions are
unavailable. Both functional differencing and approximate functional differencing
are general methods applicable to a broad class of nonlinear panel models. AOI
adapts and extends this approach to the estimation of $\mu_0$, going beyond
\cite{dhaene2023approximate} by developing the operator inversion perspective,
establishing formal results on bias decay and asymptotic normality, and providing
feasible inference procedures.

The bounds literature offers a complementary approach, deriving sharp partial 
identification regions under minimal assumptions \citep{honore2006bounds, 
chernozhukov2013average, davezies2021identification, dobronyi2021identification, 
pakel2023bounds, botosaru2024adversarial}.
In modern applications where $n$ is large and $T$ is only moderately large, these bounds 
can be very narrow, so that $\mu_0$ is nearly point-identified in practice.
AOI and bounds approaches are therefore complementary. Bounds provide robust
partial identification guarantees, while AOI provides point estimates and
inference under large-$T$ approximations.

More broadly, AOI contributes to the large literature on bias-correction in 
nonlinear panel models \citep{hahn2004jackknife, arellanobonhomme2009intlike, dhaene2015split, 
higgins2024bootstrap, bonhomme2024neyman}.
These methods rely on a preliminary consistent estimator of $A_i$. AOI does not require consistent estimation of
$A_i$, making it applicable in a broader class of settings and delivering better
finite-sample performance when $T$ is small.

\paragraph{Roadmap.}
The paper is organized as follows.
Section~\ref{sec:Setup} introduces the model, defines the target average effect,
and presents a motivating example.
Section~\ref{sec:bias-corr} formulates the estimation of average effects as an
inversion problem and shows how AOI approximately solves it.
Section~\ref{sec:estimator} defines the AOI estimator and interprets it as an
infinitely iterated bias-correction.
Section~\ref{sec:asymptotics} establishes the large-sample properties of the
estimator, including the rate double robustness result, asymptotic normality,
and feasible inference.
Section~\ref{sec:RCnumerics} presents numerical results for a random-coefficient
logit model.
Section~\ref{sec:theta} discusses inference on the common parameter $\theta_0$.
Section~\ref{sec:conclusion} concludes.
All proofs are collected in the Appendix.

\section{Setup and motivating example}
\label{sec:Setup}

This section introduces the formal framework and a motivating example.
Section~\ref{subsec:model} defines the model and the target average effect.
Section~\ref{subsec:regimeC-example} illustrates that average effects can be 
consistently estimated even when the fixed effects cannot.

\subsection{Model, average effects, and examples}\label{subsec:model}

We observe outcomes $Y_i \in \mathcal{Y}=\{y_{(k)},\ k=1,\dots,n_{\mathcal{Y}}\}$ 
and covariates $X_i \in {\cal X}\subset\mathbb{R}^{d_x}$ for units $i=1,\ldots,n$. 
We restrict attention to finite outcome sets $\mathcal{Y}$, although in principle 
the results can be extended to infinite outcome sets under appropriate conditions.
Besides the observed variables $Y_i$ and $X_i$, we allow for latent variables 
$A_i \in {\cal A}\subset\mathbb{R}^{d_a}$.
In what follows, we often drop the unit index $i$. For example, instead of $Y_i$, 
$X_i$, and $A_i$ we simply write $Y$, $X$, and $A$.

We assume that $(Y_i, X_i, A_i)$, $i=1,\ldots,n$, are independent and identically 
distributed random vectors. The model specifies the conditional outcome 
probabilities\footnote{In many panel models of interest, the conditional outcome 
probabilities also depend on a finite-dimensional common parameter $\theta_0$, so
that $\mathbb{P}(Y=y \,|\, X=x, A=\alpha) = f(y \,|\, x, \alpha, \theta_0)$ and
$\mathbb{P}(Y=y \,|\, X=x) = \int_{\mathcal{A}} f(y \,|\, x, \alpha, \theta_0) \,
\pi_0(\alpha \,|\, x) \, \dd\alpha$. When a {$\sqrt{nT}$}-consistent estimator of
$\theta_0$ is available, $\theta_0$ plays essentially no role for inference on 
$\mu_0$. To avoid carrying $\theta_0$ in the notation throughout the paper, we 
absorb it into the model and write simply $f(y \,|\, x, \alpha)$. Inference on 
$\theta_0$ itself is discussed in Section~\ref{sec:theta}.}
\begin{align}
   \mathbb{P}\left( Y= y \, \big| \, X=x, \, A = \alpha \right) 
 =  f\left(y \, \big| \, x, \alpha \right) , \qquad y \in \mathcal{Y},
    x \in \mathcal{X}, \alpha \in \mathcal{A},
    \label{model}
\end{align}
where the function $f\left(y \, \big| \, x, \alpha \right)$ is known.

Let $\pi_0(\alpha \,| \,x)$ denote the true conditional probability density function 
of $A$ given $X=x$. Integrating out $A$ from~\eqref{model}, the conditional outcome 
probabilities identified from the data are
\begin{align}
   \mathbb{P}\left( Y= y \, \big| \, X=x  \right) 
 = \int_{\cal A} \,  f\left(y \, \big| \, x, \alpha \right) \, 
 \pi_0(\alpha \,| \,x) \, \dd\alpha, \qquad y\in\mathcal{Y},x\in\mathcal{X}.
     \label{model2}
\end{align}
We impose no restrictions on $\pi_0(\alpha \,| \,x)$ nor on the marginal distribution 
of $X$, that is, we have a semiparametric model with unknown nonparametric 
component $\pi_0(\alpha \,| \,x)$.

\begin{remark}[Dynamic models]
\label{rmk:dynamic}
The setup in~\eqref{model}--\eqref{model2} accommodates certain dynamic nonlinear 
panel models. For instance, when $Y_t$ depends on $Y_{t-1}$ and $X_t$, and the 
initial condition $Y_{0}$ is observed, we set $Y = (Y_{1},\ldots,Y_{T})$ and 
$X = (Y_{0},X_{1},\ldots,X_{T})$. Our setup, however, does not allow for dynamic 
feedback from the dependent variable $Y_t$ to covariates in later periods 
$X_{t+s}, s\ge 1$.
\end{remark}

In models of the form~\eqref{model}, the primary objects of interest are often
functionals of the unknown conditional density $\pi_0(\alpha\,| \, x)$. In 
particular, we consider average effects of the form
\begin{align}
   \mu_0 &:= \mathbb{E}\left[ \mu(X,A) \right]
    = \mathbb{E}\left[ \int_{\cal A}   \mu(X, \alpha)  \,   
    \pi_0(\alpha\,|\,X) \, \dd\alpha \right] ,
    \label{eq:avg_effect}
\end{align}
where $\mu(x,\cdot)\in L^2(\mathcal{A})$ is a known function specifying 
the effect of interest, and $L^2(\mathcal{A})$ denotes the space of 
square-integrable functions on $\mathcal{A}$ with respect to Lebesgue measure.
For example, in a panel data model, the average marginal effect of the $p$-th 
regressor in period $t$ on the expected outcome in period $t$ sets
$\mu( x,\alpha) = \frac{\partial}{\partial x_{t,p}} \sum_{y \in 
\mathcal{Y}} y_t f(y \,| \, x, \alpha)$,
where $y=(y_1,\ldots,y_T)$ and $y_t$ is the $t$-th component of $y$.


\begin{example}[\textbf{Static logit}]
   \label{example:StaticLogit}
   Consider the static logit model with $T$ periods, where the common parameter
   $\theta_0$ has already been estimated (for example, using the conditional logit 
   estimator of \cite{rasch1961general,andersen1970asymptotic,chamberlain1980analysis}) 
   and absorbed into the model. We have
   $\mathcal{Y} = \{0,1\}^T$ and
   $$
    f(y \, | \,x,\alpha) = \prod_{t=1}^T
    \frac{[\exp(x_t' \beta + \alpha)]^{y_t}} {1 + \exp(x_t' \beta + \alpha)} 
    =:\prod_{t=1}^T \tilde{f}(y_t \, | \,x_t,\alpha) ,
   $$
   for $x=(x_1^\top,\dots,x_{T}^\top)^\top\in{\cal X}\subset\mathbb{R}^{KT}$, 
   $\alpha\in{\cal A}\subset \mathbb{R}$, and a known coefficient vector $\beta$.
   
   \medskip
   
   \noindent\textit{\textbf{Average treatment effect.}}
   Suppose $X_{t1}$ is binary for all $t$. Let 
   $\widetilde{X}^{(k)}=((\widetilde{X}_{1}^{(k)})^\top,\dots,
   (\widetilde{X}_{T}^{(k)})^\top)^\top$ 
   denote the counterfactual covariate vector with the first component in each 
   period set to $k\in\{0,1\}$, i.e., 
   $\widetilde{X}_{t}^{(k)}=(k,X_{t2},\dots,X_{tK})^\top$, $t=1,\dots,T$. 
   The average treatment effect of $X_{t1}$ for $t=1,\ldots,T$ on 
   $\bar{Y}:=T^{-1}\sum_{t=1}^T Y_t$ is
   \begin{align*}\mu_0
   &=\mathbb{E}\left[\sum_{y\in\mathcal{Y}}\bar{y}\int_{\cal A} \,  
   \left\{f(y \, \big| \, \widetilde{X}^{(1)}, \alpha ) 
   - f(y \, \big| \,  \widetilde{X}^{(0)}, \alpha ) 
   \right\}\, \pi_0(\alpha \,| \,X) \, \dd\alpha\right],
   \end{align*}
   where $\bar{y}=\frac1T\sum_{t=1}^T y_{t}$.
   
   \medskip

   \noindent\textit{\textbf{Average marginal effect.}}
   If, instead, $X_{t1}$ has a continuous distribution, the average marginal 
   effect of $X_{t1}$ for $t=1,\ldots,T$ on $\bar{Y}$ is 
   \begin{align*}\mu_0
  & = \mathbb{E}\left[\sum_{y\in\mathcal{Y}}\bar{y}\int_{\cal A}  
  \frac{\partial}{\partial X_{t1}} \prod_{s=1}^T \tilde{f}\left(y_{s} \, \big| \,
  X_{s}, \alpha \right) 
   \, \pi_0(\alpha \,| \,X)\,\dd\alpha\right].\end{align*}
\end{example}

\subsection{Random-coefficient binary choice model}\label{subsec:regimeC-example}

Consider a static binary choice panel model with $T$ periods and 
$\mathcal{Y} = \{0,1\}^T$.
There is no common parameter $\theta_0$, and the unit-specific parameters are 
$A_i = (A_{i1}, A_{i2})^\top \in \mathcal{A} \subset \mathbb{R}^2$, where 
$A_{i1}$ is an additive effect and $A_{i2}$ is a slope.
The outcome of unit $i$ in period $t$ is
\[
Y_{it} = \mathbf{1}(A_{i1} + X_{it} A_{i2} + U_{it} \geq 0),
\]
where $X_{it} \in \mathbb{R}$ is a scalar covariate and $U_{it}$ is 
independently and identically distributed across $i$ and $t$ with distribution 
function $F$, and independent of $A_i$.
The conditional outcome probabilities are
$$
f(y \,|\, x, \alpha) = \prod_{t=1}^T 
[F(\alpha_1 + x_t \alpha_2)]^{y_t}
[1 - F(\alpha_1 + x_t \alpha_2)]^{1-y_t},
\qquad y \in \mathcal{Y},
x \in \mathcal{X}, \alpha \in \mathcal{A},
$$
where $x = (x_1, \dots, x_T)^\top \in \mathcal{X} \subset \mathbb{R}^T$ and 
$\alpha = (\alpha_1, \alpha_2)^\top \in \mathcal{A} \subset \mathbb{R}^2$.
The average effect of interest is
\[
\mu_0 = \mathbb{E}[F(A_{i1} + A_{i2})],
\]
the average predicted probability when the covariate is set to one. If, for all $i$,
$X_{it} = 1$ for some period $t$, then $F(A_{i1} + A_{i2}) = 
\mathbb{P}(Y_{it} = 1 \,|\, A_i, X_{it}=1)$, and $\mu_0 = \mathbb{E}[Y_{it}]$ is 
identified from the data without any parametric assumption on $F$ 
\citep[see][]{chernozhukov2013average}. When $X_{it} \neq 1$ for all $t$, 
however, evaluating $F(A_{i1} + A_{i2})$ requires extrapolating from the 
observed covariate values to the target value one. The parametric form of $F$ 
then becomes essential, and the difficulty of the problem depends on how much 
information the observed covariate variation provides about this extrapolation.

 Consistent estimation of 
$A_i = (A_{i1}, A_{i2})$ as $T \to \infty$ requires the covariates to exhibit 
sufficient variation across periods. A necessary and sufficient condition is
\begin{align}
\sum_{t=1}^T (X_{it} - \bar{X}_{iT})^2 \to \infty,
    \label{SufficientVariationXit}
\end{align}
where $\bar{X}_{iT} = T^{-1} \sum_{t=1}^T X_{it}$. When this $L^2$ condition 
on the covariate variation fails, the Fisher information that the data provide 
about $A_i$ remains bounded as $T \to \infty$, and $A_i$ cannot be consistently 
estimated.

Consistent estimation of the average effect $\mu_0$ requires weaker assumptions on the covariate variation. {Under
appropriate regularity conditions on the model primitives (bounded
covariates and a compact parameter space for $A_i$),} a sufficient condition is
\begin{align}
\sum_{t=1}^T |X_{it} - \bar{X}_{iT}| \to \infty.
    \label{SufficientVariationMu}
\end{align}
This is an $L^1$ condition on the covariate variation, and it is strictly weaker
than~\eqref{SufficientVariationXit}:
condition~\eqref{SufficientVariationXit} always
implies~\eqref{SufficientVariationMu}, since
$\sum_{t=1}^T |X_{it} - \bar{X}_{iT}| \geq
\bigl(\sum_{t=1}^T (X_{it} - \bar{X}_{iT})^2\bigr)^{1/2}$, but the converse
does not hold in general.

The gap between conditions~\eqref{SufficientVariationXit} 
and~\eqref{SufficientVariationMu} defines the territory of regime (C) from the 
introduction: settings in which $\mu_0$ can be consistently estimated even though 
$A_i$ cannot.
To make these conditions concrete, consider a two-block covariate design where
\[
X_{it} = 0 \text{ for } t \leq T/2, \qquad X_{it} = c_T \text{ for } t > T/2,
\]
with $T$ even and $c_T > 0$. Here $c_T$ is allowed to depend on $T$, reflecting
a triangular array asymptotic framework in which the covariate design may change
as $T$ grows. In this design,
condition~\eqref{SufficientVariationXit} reduces to $Tc_T^2 \to \infty$ and
condition~\eqref{SufficientVariationMu} reduces to $Tc_T \to \infty$.
With a slight abuse of notation, define the within-unit sums
\[
Y_{i1} = \sum_{t=1}^{T/2} Y_{it}, \qquad Y_{i2} = \sum_{t=T/2+1}^{T} Y_{it}.
\]
These two sums are jointly sufficient statistics for $A_i$ given $X_i$, with
$Y_{i1} \sim \mathrm{Bin}(T/2,\, F(\alpha_1))$ and
$Y_{i2} \sim \mathrm{Bin}(T/2,\, F(\alpha_1 + c_T\,\alpha_2))$ conditionally
on $A_i = \alpha$.

\paragraph{Regime (B): $c_T$ fixed.}
When $c_T = c$ is a fixed positive constant that does not depend on $T$, both
conditions are satisfied and $A_i$ can be consistently estimated as
$T \to \infty$. The standard approach would estimate each $A_i$ by maximum
likelihood and form the plug-in estimator
$\hat{\mu}_{\rm FE} = n^{-1} \sum_{i=1}^n F(\hat{A}_{i1} + \hat{A}_{i2})$.
Due to the incidental parameter problem, this estimator has a bias of order
$O(1/T)$ \citep{hahn2004jackknife}. By contrast, the method developed in this
paper achieves a bias converging to $0$ at the exponential rate $\rho^T$ for some
$0 < \rho < 1$, which is much faster. This illustrates that our approach offers
substantial improvements over standard methods even within regime (B).

\paragraph{Regime (C): $c_T = 1/\sqrt{T}$.}
When $c_T = 1/\sqrt{T}$, we have $Tc_T^2 = 1$ and $Tc_T = \sqrt{T}$,
so~\eqref{SufficientVariationXit} fails
but~\eqref{SufficientVariationMu} holds. The fixed effect $A_i$ cannot be
consistently estimated, whereas $\mu_0$ can. Standard bias-correction methods,
which rely on a preliminary consistent estimator of $A_i$, are not applicable
in this setting. Our method remains applicable because it operates at the
distributional level and does not require $A_i$ to be consistently estimated.

A consistent estimator of $\mu_0$ is
\[
\hat{\mu}_0 = \frac{1}{n} \sum_{i=1}^n m_T(Y_{i1}, Y_{i2}),
\]
where $m_T: \{0,\ldots,T/2\}^2 \to \mathbb{R}$ is a function of the two 
within-unit sums that can be chosen to make $\hat{\mu}_0$ consistent for 
$\mu_0$ as $n, T \to \infty$.
For the standard logistic $F$, an explicit such choice is
\[
m_T(y_1, y_2) = \sum_{m=0}^{T/2}\sum_{l=0}^{T/2}
g_T(t_m, t_l)\, L_m(y_1)\, L_l(y_2),
\]
where $g_T(p_1, p_2) = F\bigl((1-\sqrt{T})\,\mathrm{logit}(p_1) +
\sqrt{T}\,\mathrm{logit}(p_2)\bigr)$, $t_0,\ldots,t_{T/2}$ are Chebyshev
nodes on $[\epsilon, 1-\epsilon]$ {for any fixed $\epsilon \in (0,1/2)$},
and $L_0,\ldots,L_{T/2}$ are the
associated unbiased polynomial estimators.
With this choice, the bias satisfies
\[
\sup_{\alpha \in \mathcal{A}}
\bigl|\mathbb{E}[m_T(Y_{i1}, Y_{i2}) \mid A_i = \alpha] - F(\alpha_1 + \alpha_2)\bigr|
\leq C\, e^{-c_0 \sqrt{T}}
\]
for constants $C, c_0 > 0$, so the bias decays to zero faster than any
polynomial rate in $T$, even though $A_i$ is not consistently estimable.
Full details and the proof are given in 
Appendix~\ref{app:proof-lemma1}.

\section{Average effect estimation as an inversion problem}\label{sec:bias-corr}

The average effect $\mu_0 = \mathbb{E}[\mu(X,A)]$ depends on the distribution
of the unobserved $A$. We only observe outcomes $Y$, whose distribution is
determined by the distribution of $A$ through the model. This section explains
how to estimate $\mu_0$ by approximately inverting this relationship. The
covariate value $x \in \mathcal{X}$ is fixed throughout and plays no essential
role. Where it aids readability, we suppress $x$ from the discussion.

The model $f(y \,|\, x, \alpha)$ tells us the probability of each outcome $y$
given a particular fixed-effect value $\alpha$. But we are not interested in
$f$ as a function of individual fixed effects. What matters is the mapping that
$f$ induces at the level of \emph{distributions}: if the fixed effect $A$ has
distribution $\pi$ on $\mathcal{A}$, then integrating $f$ against $\pi$
produces the distribution of the outcome $Y$ on $\mathcal{Y}$. Formally, for
any distribution $\pi(\alpha \,|\, x)$ of $A$ given $X = x$,
\begin{equation}\label{eq:mapping}
\mathbb{P}(Y=y\,| \, X=x) = \int_{\mathcal{A}} f(y\,| \, x,\alpha)\,
\pi(\alpha\,| \, x)\,\dd\alpha,
\qquad  y\in\mathcal{Y}.
\end{equation}
This is a linear map from distributions on $\mathcal{A}$ to distributions on
$\mathcal{Y}$. Its input is a density $\pi(\cdot \,|\, x)$ on $\mathcal{A}$,
which is an infinite-dimensional object. Its output is a vector of
$n_{\mathcal{Y}} = |\mathcal{Y}|$ outcome probabilities, which is
finite-dimensional for any given $T$.

The data identify the output of~\eqref{eq:mapping}, that is, the outcome
probabilities $\mathbb{P}(Y = y \,|\, X = x)${, while} the average effect $\mu_0$
depends on the input, the true distribution $\pi_0(\cdot \,|\, x)$. Estimating
$\mu_0$ therefore {corresponds to} inverting the distributional
mapping~\eqref{eq:mapping}: given the observed output, recover enough about the
input distribution to compute $\mu_0 = \mathbb{E}[\mu(X,A)]$.

{Unfortunately,} this is not possible in general. The input is infinite-dimensional
and the output is finite-dimensional, so many different input distributions
produce the same output. The average effect $\mu_0$, which depends on $\pi_0$
through~\eqref{eq:avg_effect}, is therefore typically not point-identified for
finite $T$.

\subsection*{Approximate inversion}

Although the full inversion fails, it succeeds when the input distribution is
restricted to a finite-dimensional class. Let $\phi_x^{(j)}(\alpha)$,
$j=1,\dots,J$, be a collection of basis functions and consider the linear span
\[
\Pi_{\phi_x} := \left\{\pi_x:\mathcal{A}\to\mathbb{R}\ \big|\
\pi_x(\alpha)=\sum_{j=1}^{J} \nu_x^{(j)}\,\phi_x^{(j)}(\alpha),\
\nu_x^{(j)}\in\mathbb{R}\right\}.
\]
If $\pi_0(\cdot\,|\,x)$ lies in $\Pi_{\phi_x}$, the dimensions match: the
mapping~\eqref{eq:mapping} becomes a linear system with $n_{\mathcal{Y}}$
equations and $J \leq n_{\mathcal{Y}}$ unknowns. Under the rank condition that
the $n_{\mathcal{Y}}\times J$ matrix with entries
$\bigl[\int_{\mathcal{A}} f(y_{(k)}\,| \, x,\alpha)\,
\phi_x^{(j)}(\alpha)\,\dd\alpha\bigr]_{k,j}$
has rank $J$, this system is invertible and the distribution
$\pi_0(\cdot\,|\,x)$ can be recovered from the observed outcome probabilities.

The inversion yields an estimating function
$m_T(\cdot,x):\mathcal{Y}\to\mathbb{R}$ that satisfies, \emph{for any
$\pi_0(\cdot\,|\,x) \in \Pi_{\phi_x}$},
\begin{equation}\label{eq:exact_moment}
\mathbb{E}[m_{T}(Y,x)\,| \, A=\alpha,X=x]=\mu(x,\alpha),
\quad \text{for all } \alpha\in\mathcal{A}.
\end{equation}
This condition holds for every value of $\alpha$, so integrating both sides
over $A$ eliminates the fixed effect entirely:
$\mathbb{E}[m_T(Y,x)\,|\,X=x] = \mathbb{E}[\mu(x,A)\,|\,X=x]$. The
left-hand side depends only on the observed data. When~\eqref{eq:exact_moment}
holds, $\mu_0$ is identified and can be estimated without any knowledge of the
distribution $\pi_0$ or of the individual $A_i$.

The restriction $\pi_0(\cdot\,|\,x)\in\Pi_{\phi_x}$ is, of course, too strong
to maintain as an assumption. AOI uses $\Pi_{\phi_x}$ as an
\emph{approximation device}. Even when $\pi_0(\cdot\,|\,x)$ does not lie in
$\Pi_{\phi_x}$, there is a unique element $\pi_x^*\in\Pi_{\phi_x}$ whose image
under~\eqref{eq:mapping} matches the observed outcome probabilities, that is,
$\pi_x^*$ uniquely solves
\[
\mathbb{P}(Y=y\,| \, X=x)=\int_{\mathcal{A}} f(y\,| \, x,\alpha)\,
\pi_x^*(\alpha)\,\dd\alpha, \qquad y\in\mathcal{Y}.
\]
AOI estimates $\mathbb{E}[\mu(x,A)]$ by
$\mathbb{E}_{\pi_x^*}[\mu(x,A)]$, the average-effect functional evaluated at
$\pi_x^*$ instead of $\pi_0(\cdot\,|\,x)$. Since $\pi_x^*$ only approximates
$\pi_0(\cdot\,|\,x)$, the exact condition~\eqref{eq:exact_moment} no longer
holds. Over all distributions $\pi_0$, we can only achieve the approximation
\begin{equation}\label{eq:approx_moment}
\mathbb{E}[m_{T}(Y,x)\,| \, A=\alpha,X=x]\approx\mu(x,\alpha),
\quad \text{for all } \alpha\in\mathcal{A},
\end{equation}
with an error that depends on how well $\Pi_{\phi_x}$ captures
$\pi_0(\cdot\,|\,x)$.

The approximation error in~\eqref{eq:approx_moment} can be made small because
the dimension of $\Pi_{\phi_x}$ can grow with $T$. This is possible because the
number of distinct outcome values $n_{\mathcal{Y}} = |\mathcal{Y}|$ itself
grows with $T$.\footnote{For instance, in a binary-outcome panel model with $T$
periods, $n_{\mathcal{Y}}=2^T$.} As $T$ increases, there are more outcome
probabilities available to pin down the input distribution, so the approximating
space $\Pi_{\phi_x}$ can be made richer while the rank condition continues to
hold. Under appropriate smoothness conditions on $\mu$ and $\pi_0$, this means
the approximation error in~\eqref{eq:approx_moment} vanishes rapidly as $T$
grows. The formal rates are established in Section~\ref{sec:asymptotics}.

A central feature of the procedure described above is that it operates entirely
at the level of distributions. The distributional mapping~\eqref{eq:mapping}
takes as input the distribution of $A$ and produces as output the distribution
of $Y$. Inverting this mapping recovers information about the distribution of
$A$, which is all that is needed to compute $\mu_0 = \mathbb{E}[\mu(X,A)]$. At
no point does the method construct an estimator of the fixed effect $A_i$ for
any individual unit $i$. All that is needed is the aggregate distribution of
outcomes, not the ability to trace outcomes back to individual fixed-effect
values.

This is a fundamental difference from standard bias-correction methods, which
first estimate each $A_i$ and then correct the resulting bias. Those methods
require $A_i$ to be consistently estimable, which places them in regime (B)
of the classification in Section~\ref{sec:introduction}. Because AOI bypasses
individual fixed-effect estimation, it applies equally in regimes (B) and (C).
Whether or not $A_i$ can be consistently estimated is irrelevant to the
procedure. This is the reason why AOI remains valid in settings where existing
methods fail.

The approach is related to functional differencing
\citep{bonhomme2012functional}, which constructs exact moment conditions of the
form~\eqref{eq:exact_moment} that hold for \emph{any} distribution $\pi_0$
\citep[see also][for applications to average marginal
effects]{dano2023transition, aguirregabiria2024identification}. Exact moment
conditions of this kind exist only in specific models and for specific choices
of $\mu_0$. AOI works with the approximate moment
conditions~\eqref{eq:approx_moment} instead, which are broadly available, at
the cost of a bias that vanishes as $T$ grows.

\section{The AOI estimator}
\label{sec:estimator}

This section presents the AOI estimator. 
Section~\ref{subsec:prior-posterior} introduces the prior and posterior 
densities that serve as building blocks. Section~\ref{subsec:transition} 
defines the transition matrix $Q(x)$, with entries given by posterior predictive probabilities,
and its pseudoinverse. 
Section~\ref{subsec:def-estimator} presents the estimator itself and 
characterizes the sieve space it implicitly uses. 
Section~\ref{sec:iterated} shows that the estimator arises as the limit 
of an iterated bias correction.

\subsection{Prior, posterior, and predictive densities}
\label{subsec:prior-posterior}

The AOI estimator is built from a user-chosen prior density 
$\pi_{\rm prior}(\alpha\,| \, x)$ for $A$ given $X=x$. This prior is not a 
belief about $A$ in the Bayesian sense; rather, it is a computational device 
that determines the function space $\Pi_{\phi_x}$ used to approximate the unknown true density 
$\pi_0(\cdot\,| \, x)$. As we show in Section~\ref{subsec:def-estimator}, the 
choice of prior pins down functions that span $\Pi_{\phi_x}$,
but the AOI estimator remains consistent as $T\to\infty$ for any prior satisfying the 
conditions below.

We require that $\pi_{\rm prior}(\alpha\,| \, x)$ integrates to one and satisfies 
the following positivity condition:
\begin{align}
     \pi_{\rm prior}\left(\alpha \,| \,x \right) > 0 , \qquad 
     \text{for all $\alpha \in {\cal A}$ and $x \in {\cal X}$.}
     \label{ConditionPrior}
\end{align}
The prior need not depend on $x$: choosing 
$\pi_{\rm prior}(\alpha\,| \, x) = \pi_{\rm prior}(\alpha)$ simplifies 
computation, but $x$-dependent priors are permitted. The prior  
$\pi_{\rm prior}$ may differ from the true $\pi_0$.

Given $\pi_{\rm prior}$, the posterior density of $A$ conditional on $Y=y$
and $X=x$ follows from Bayes' rule:
\begin{align}
   \pi_{\rm post}(\alpha \,|\, y,x) 
     = \frac{ f(y \, | \,x,\alpha) \, \pi_{\rm prior}(\alpha\,|\,x)} 
     {p_{\rm prior}(y \, | \, x)} \, ,
     \label{DefPost}
\end{align}     
where 
\begin{align}
p_{\rm prior}(y \, |\, x) 
= \int_{\cal A} f( y \, | \,x,\alpha) \, 
\pi_{\rm prior}( \alpha\,|\,x) \, \dd\alpha
\label{eq:prior_predictive}
\end{align}
is the prior predictive probability of outcome $y$. We assume that the prior 
is chosen such that
\begin{align}
     p_{\rm prior}(y \, |\, x) > 0 , \qquad \text{for all   
     $y\in\mathcal{Y}$ and $x \in {\cal X}$.}
     \label{ConditionPrior2}
\end{align}
Condition~\eqref{ConditionPrior2} ensures that the posterior~\eqref{DefPost} 
is well-defined for every possible outcome $y\in\mathcal{Y}$. It is automatically satisfied when 
$\pi_{\rm prior}$ satisfies~\eqref{ConditionPrior} and $\mathcal{A}$ has 
positive Lebesgue measure.

\subsection{The transition matrix and its pseudoinverse}
\label{subsec:transition}

Given $x \in {\cal X}$, the posterior predictive 
probability of a ``future" outcome $\widetilde y\in {\cal Y}$ after having 
observed $y \in {\cal Y}$ is
\begin{align}
    Q( \widetilde y \, | \, y,x) = \int_{\cal A} f \left(\widetilde y 
    \, \big| \, x,\alpha \right) \pi_{\rm post} \left( \alpha \, 
    \big| \, y,x \right) \dd\alpha \, .
    \label{DefQ}
\end{align}
Collecting these probabilities into a matrix, let $Q(x)$ be the 
$n_{\cal Y} \times n_{\cal Y}$ matrix with entries
$Q_{k,\ell}(x) = Q(y_{(k)} \, | \, y_{(\ell)},x)$.
Thus $Q(x)$ is a transition matrix: its $(k,\ell)$-entry is the 
probability that an independent replicate of $Y$ equals $y_{(k)}$, given 
that the original observation was $y_{(\ell)}$ and that the prior 
$\pi_{\rm prior}$ is used to form beliefs about $A$. We have the following lemma from \cite{dhaene2023approximate}.

\begin{lemma}
    \label{lemma:MatrixQ}
     Let $x \in {\cal X}$.
     Assume that $p_{\rm prior}(y \, |\, x) > 0$ for all $y \in {\cal Y}$.
     Then $Q(x)$ is diagonalizable and all its eigenvalues are real 
     numbers in the interval $[0,1]$.
\end{lemma}

Let $\lambda_1(x) \ge \ldots \ge 
\lambda_{n_{\mathcal Y}}(x)$ denote the eigenvalues of 
$Q(x)$, ordered in descending order, and let
 $U(x)$ be the $n_{\mathcal Y}\times n_{\mathcal Y}$ 
matrix whose columns are the corresponding right eigenvectors. 
Define
$ \Lambda(x) := \operatorname{diag}\bigl(\lambda_k(x)\bigr)_{k=1}^{n_{\mathcal Y}}. $
By Lemma~\ref{lemma:MatrixQ}, 
$\lambda_k(x)\in[0,1]$ for all $k$, and
$$
Q(x)
=
U(x)\,\Lambda(x)\,U^{-1}(x).
$$
Our definition of the AOI estimator involves the inverse $Q(x)^{-1}$, or, when $Q(x)$ is singular, its Drazin inverse,
\[
Q(x)^D
:=
U(x)\,\Lambda(x)^D\,U(x)^{-1},
\]
where
$ \Lambda(x)^D$ is obtained from $ \Lambda(x)$ by inverting the nonzero eigenvalues $\lambda_k(x)$ and setting the zero eigenvalues to zero. If $Q(x)$ is nonsingular, then $Q(x)^D=Q(x)^{-1}$.
In general, when $Q(x)$ is singular, the Drazin inverse differs from the Moore-Penrose pseudoinverse, unless $Q(x)$ is symmetric, which is usually not the case in our setup.

\subsection{Definition of the estimator}
\label{subsec:def-estimator}

Suppose we are interested in estimating the average effect 
$\mu_0=\mathbb{E}[\mu(X,A)]$. The AOI estimator is
\begin{align}
   \widehat \mu^{(\infty)} &:= \frac 1 n \sum_{i=1}^n 
   w^{(\infty)} ( Y_i,X_i),
   \label{eq:AOI_estimator}
\end{align}
where the estimating function $w^{(\infty)}$ is defined as
\begin{align}
    w^{(\infty)}(y,x) 
  &:= \sum_{\tilde y \in {\cal Y}} 
   \left(\int_{\cal A} \mu( x, \alpha ) \,    
   \pi_{\rm post}(\alpha \,|\, \tilde y,x) \, \dd\alpha\right)
   \bigl\{Q(x)^D\bigr\}_{\tilde y,y}.
   \label{eq:Winf}
\end{align}
The estimating function has a transparent structure. For each hypothetical 
outcome $\tilde{y}$, the integral computes the posterior mean of $\mu(x,A)$. These posterior means are 
then reweighted by the entries of the Drazin inverse $Q(x)^D$, 
which corrects for the distortion introduced by using the posterior rather 
than the true conditional distribution of $A$.

The estimator uses the approximate moment condition
\[
\mathbb{E}\!\left[w^{(\infty)}(Y,X)
-\mu(X,A)\right]\approx 0,
\]
where the approximation becomes exact under the conditions given below.

As formally established in Theorem~\ref{theorem.inf}, the AOI estimator is exactly unbiased whenever, 
for all $x \in {\cal X}$, there exists 
$\nu(\cdot,x):\mathcal{Y}\to\mathbb{R}$ such that
\begin{align}
   \pi_0(\alpha\,| \, x) = \pi_{\rm prior}(\alpha\,|\,x)
      \sum_{y \in {\cal Y}} \nu(y,x) \, f(y \,| \, x,\alpha), \qquad \text{for all } \alpha\in\mathcal{A}.
      \label{eq:exact_unbiased}
\end{align}
In the notation of Section~\ref{sec:bias-corr}, condition~\eqref{eq:exact_unbiased} 
means that $\pi_0(\cdot\,| \, x)$ lies in the sieve space $\Pi_{\phi_x}$ with 
basis functions
\[
\phi_x^{(k)}(\alpha) := \pi_{\rm prior}(\alpha\,|\,x)\,
f(y_{(k)}\,| \, x,\alpha), 
\qquad k=1,\dots,n_{\mathcal{Y}}.
\]
This particular choice of sieve space -- prior-weighted likelihood 
functions -- has two desirable properties:
\begin{enumerate}[(i)]
\item It enables the interpretation of our method as a bias correction that is iterated infinitely many times, as developed in 
Section~\ref{sec:iterated}.
\item It guarantees that $\widehat\mu^{(\infty)}$ is exactly unbiased
whenever $\mu_0$ is a linear combination of the 
outcome probabilities $\mathbb{P}(Y=y\,|\,X=x)$, $y\in\mathcal{Y}$, $x\in\mathcal{X}$.\footnote{Or, more primitively, that, for all $x\in\mathcal{X}$, $\mu(x,\cdot)$ is a linear combination of the 
likelihood functions $f(y\,| \, x,\cdot)$, $y\in\mathcal{Y}$.} This is natural, 
since such average effects are point-identified.
\end{enumerate}

\subsection{Interpretation as iterated bias correction}
\label{sec:iterated}

We now show that $w^{(\infty)}$ arises as the limit of an iterated 
bias correction, extending the approach of 
\citet{dhaene2023approximate} to average effects.
The starting point is the plug-in estimating function
\begin{equation}
\label{eq:w0}
w^{(0)}(y,x)
=
\int_{\mathcal{A}} \mu(x,\alpha)\,
\pi_{\rm post}(\alpha\,| \, y,x)\,\dd\alpha,
\end{equation}
which replaces the unknown $\pi_0(\alpha\,| \, x)$ by the posterior
$\pi_{\rm post}(\alpha\,| \, y,x)$. This introduces a bias
\begin{align}
   \mathbb{E}\!\left[w^{(0)}(Y,X)-\mu(X,A)\right]
   &=  \mathbb{E}\!\left[\widetilde\mu^{(0)}(X,A)\right],
   \label{eq:plugin_bias}
\end{align}
where
$      \widetilde\mu^{(0)}(x,\alpha) := \sum_{y\in \cal Y}
       w^{(0)}(y,x)\, f(y\,| \, x,\alpha) - \mu(x,\alpha),
$
see Appendix~\ref{app:plugin_bias} for details. The right-hand side
of~\eqref{eq:plugin_bias} has the same form as the average effect
$\mu_0 = \mathbb{E}[\mu(X,A)]$, but with $\mu$ replaced by the conditional
bias function $\widetilde\mu^{(0)}$. The bias of $w^{(0)}$ is therefore itself
an average effect, and can be estimated by the same plug-in construction.
Replacing $\pi_0$ by $\pi_{\rm post}$ in $\mathbb{E}[\widetilde\mu^{(0)}(X,A)]$
yields the estimating function
\[
   b^{(0)}(y,x) := \int_{\cal A} \widetilde\mu^{(0)}(x,\alpha)\,
   \pi_{\rm post}(\alpha\,|\,y,x)\,\dd\alpha,
\]
which, by the definition of $Q(x)$ in~\eqref{DefQ}, simplifies to
$b^{(0)}(\cdot,x) = Q(x)\,w^{(0)}(\cdot,x) - w^{(0)}(\cdot,x)$.
Subtracting
this estimated bias from $w^{(0)}$ gives the bias-corrected estimating
function
\[
w^{(1)}(\cdot,x) = w^{(0)}(\cdot,x) - b^{(0)}(\cdot,x)
= \bigl(2\,\mathbb{I}_{n_{\cal Y}} - Q(x)\bigr)\,w^{(0)}(\cdot,x),
\]
where $\mathbb{I}_{n_{\cal Y}}$ denotes the $n_{\mathcal{Y}}\times n_{\mathcal{Y}}$ identity matrix.
This bias correction idea can be iterated: The bias of $w^{(q)}$ takes the
analogous form $\mathbb{E}[\widetilde\mu^{(q)}(X,A)]$ with conditional bias
function
\[
\widetilde\mu^{(q)}(x,\alpha) := \sum_{y\in \cal Y}
       w^{(q)}(y,x)\, f(y\,|\,x,\alpha) - \mu(x,\alpha),
\]
which is again an average-effect function. Estimating its expectation by
the same plug-in construction produces
$b^{(q)}(\cdot,x) = Q(x)\,w^{(q)}(\cdot,x) - w^{(0)}(\cdot,x)$ --- the second
term is just the plug-in $\int_{\cal A}\mu(x,\alpha)\pi_{\rm post}(\alpha\,|\,y,x)\dd\alpha
= w^{(0)}(y,x)$, which does not depend on $q$. The next estimating function
is $w^{(q+1)} := w^{(q)} - b^{(q)}$, and rearranging gives the simple
recursion
\begin{equation*}
w^{(q+1)}(\cdot,x) = w^{(0)}(\cdot,x) +
\bigl(\mathbb{I}_{n_{\cal Y}} - Q(x)\bigr)\,w^{(q)}(\cdot,x), \qquad q=0,1,2,\ldots.
\end{equation*}
which yields the closed-form expression
\begin{align}
     w^{(q)}(y,x) &=
     \sum_{\tilde y \in {\cal Y}}
   \left(\int_{\cal A} \mu( x, \alpha ) \,
   \pi_{\rm post}(\alpha \,|\, \tilde y,x) \, \dd\alpha\right)
   \left\{ \sum_{r=0}^q  \bigl[\mathbb{I}_{n_{\cal Y}} - Q(x)
     \bigr]^r \right\}_{\tilde y,y}.
     \label{eq:wq}
\end{align}
The corresponding estimator is 
   $$\widehat \mu^{(q)} := \frac 1 n \sum_{i=1}^n 
   w^{(q)} ( Y_i,X_i).$$
The partial sum 
$\sum_{r=0}^q [\mathbb{I}_{n_{\cal Y}} - Q(x)]^r$ is a truncated 
Neumann series for $Q(x)^D$. The following lemma confirms that 
this series converges (in the appropriate sense) to $Q(x)^D$ as $q\to\infty$,
so that $w^{(q)}$ converges to $w^{(\infty)}$.

\begin{lemma}\label{lemma.qinf}
    For all $y\in{\cal Y}$ and $x\in{\cal X}$,
    \begin{align*}
          w^{(\infty)}(y,x) &=
     \lim_{q \rightarrow \infty}  w^{(q)}(y,x).
    \end{align*}
\end{lemma}

\noindent
The convergence holds because the eigenvalues of 
$\mathbb{I}_{n_{\cal Y}}-Q(x)$ lie in $[0,1)$ for all nonzero 
eigenvalues of $Q(x)$, so the geometric series converges on the 
range of $Q(x)$. Eigenvalues equal to zero contribute divergent 
terms, but these lie in the kernel of $Q(x)$ and are annihilated 
when multiplied by the posterior means.

\section{Asymptotic theory}
\label{sec:asymptotics}

This section establishes the large-sample properties of the AOI estimator.
We study the estimator $\widehat{\mu}^{(\infty)}$ defined in \eqref{eq:AOI_estimator}--\eqref{eq:Winf}.
The analysis proceeds in two steps. Section~\ref{subsec:bias} characterizes 
the bias $\mu_*^{(\infty)}-\mu_0$, where 
$\mu_*^{(\infty)}:=\mathbb{E}[w^{(\infty)}(Y,X)]$, and establishes 
rates of convergence. Section~\ref{sec:InferenceKnownTheta} provides 
asymptotic normality and feasible inference under asymptotics where 
$n,T\to\infty$ jointly.

\subsection{Bias}\label{subsec:bias}

\subsubsection{Characterization of the bias}\label{subsubsec:bias-char}

The bias of $\widehat{\mu}^{(\infty)}$ admits a clean decomposition in 
terms of two approximation errors. Fix $x \in\mathcal{X}$ and define the following two 
subspaces of $L^2(\mathcal{A})$.

\paragraph{Likelihood span.}
Let $\mathcal{F}(x)$ denote the subspace of $L^2(\mathcal{A})$ spanned by the likelihood functions 
$\{f(y\,| \, x,\cdot):y\in\mathcal{Y}\}$, and write
\[
\mu(x,\cdot)=\mu_{\mathcal{F}(x)}(x,\cdot) + 
\mu_{\mathcal{F}(x)^\perp}(x,\cdot),
\]
where $\mu_{\mathcal{F}(x)}$ and $\mu_{\mathcal{F}(x)^\perp}$ denote the 
orthogonal projections of $\mu(x,\cdot)$ onto $\mathcal{F}(x)$ and 
its complement, respectively. The residual 
$\mu_{\mathcal{F}(x)^\perp}$ measures how well the likelihood basis 
approximates the average-effect functional.

\paragraph{Prior-weighted likelihood span.}
Let $\mathcal{P}(x)$ denote the subspace of $L^2(\mathcal{A})$ spanned by the functions
$\phi_x^{(k)}(\alpha)=\pi_{\rm prior}(\alpha\,|\,x)
f(y_{(k)}\,| \, x,\alpha)$, 
$k=1,\dots,n_{\mathcal{Y}}$, and write
\[
\pi_0(\cdot\,| \, x)=\pi_{\mathcal{P}(x)}(\cdot\,| \, x) + 
\pi_{\mathcal{P}(x)^\perp}(\cdot\,| \, x),
\]
where $\pi_{\mathcal{P}(x)}$ and $\pi_{\mathcal{P}(x)^\perp}$ are the 
projections of $\pi_0(\cdot\,| \, x)$ onto $\mathcal{P}(x)$ and its complement. The residual 
$\pi_{\mathcal{P}(x)^\perp}$ measures how well linear combinations of the prior-weighted likelihood 
functions approximate the true fixed-effect distribution.

The following theorem states our main result on bias.

\begin{theorem}\label{theorem.inf} Suppose that $\cal A$ is compact and the prior is uniform, i.e., 
   $\pi_{\rm prior}(\alpha\,|\,x)=(\int_\mathcal{A}1 \dd\alpha )^{-1}$ for all $\alpha\in\mathcal{A}$ and $x\in\mathcal{X}$. Then
   \[
   \mu_*^{(\infty)}-\mu_0 = -\,\mathbb{E}\!\left[\int_{\mathcal{A}}
   \mu_{\mathcal{F}(X)^\perp}(X,\alpha)\,
   \pi_{\mathcal{P}(X)^\perp}(\alpha\,| \, X)\,{\rm{d}}\alpha\right].
   \]
\end{theorem}

\begin{remark}[Non-uniform priors]
\label{rmk:nonuniform}
The compactness and the uniform-prior assumption in Theorem~\ref{theorem.inf} are without loss of 
generality. For a non-compact $\cal A$ or a non-uniform prior, one can apply a change of variables that 
renders the prior uniform on $[0,1]^{d_a}$ and then invoke Theorem~\ref{theorem.inf} in the 
transformed model. When $d_a=1$, define $\bar{A}=\Pi_{\rm prior}(A\,| \, X)$, 
where 
$\Pi_{\rm prior}(\alpha\,| \, x):=\int_{-\infty}^\alpha 
\pi_{\rm prior}(a\,| \, x)\,{\rm{d}}a$ 
is the prior CDF. The transformed fixed effect $\bar{A}$ is uniform on 
$[0,1]$, and the original model is observationally equivalent to the model with conditional outcome probabilities
\[
{\rm Pr}(Y=y\,|\,X=x,\,\bar{A}=\bar{\alpha}) 
= f(y\,|\,x,\,\Pi_{\rm prior}^{-}(\bar{\alpha}\,| \, x)),
\]
where $\Pi_{\rm prior}^{-}(u\,| \, x)
:=\inf\{\alpha\in\mathbb{R}:u\le\Pi_{\rm prior}(\alpha\,| \, x)\}$ is the conditional quantile function. 
The average effect becomes 
$\mu_0=\mathbb{E}[\mu(X,\Pi_{\rm prior}^{-}(\bar{A}\,| \, X))]$. 
In the multidimensional case ($d_a>1$), a sequential 
conditioning scheme as in \cite{rosenblatt1952remarks} can be used to
transform a non-uniform prior into a uniform prior on $[0,1]^{d_a}$.
\end{remark}

\subsubsection{Discussion of the bias structure}

Theorem~\ref{theorem.inf} decomposes the bias as the inner product of two 
approximation residuals. We discuss each component and then derive the implied 
rates.

\paragraph{(i) Approximation error to $\mu(X,\cdot)$.}
The factor $\mu_{\mathcal{F}(X)^\perp}(X,\alpha)$ is the component 
of $\mu(X,\cdot)$ orthogonal to the likelihood span 
$\mathcal{F}(X)$. It is small when $\mu(X,\cdot)$ is well 
approximated by linear combinations of the likelihood functions 
$f(y\,| \, X,\cdot)$, $y\in\mathcal{Y}$. In particular, 
Theorem~\ref{theorem.inf} implies exact unbiasedness when 
$\mu(x,\cdot)\in\mathcal{F}(x)$ for every $x\in\cal X$: the average effect 
$ \mathbb{E}\left[ \mu(x,A) \right]$
is then a 
linear combination of the outcome probabilities $\mathbb{P}(Y=y\,|\,X=x),y\in\cal Y$, and is 
therefore point-identified. The residual norm
$\|\mu_{\mathcal{F}(X)^\perp}\|_{L^2(\mathcal{A})}$ can be made small when $\mu(X,\cdot)$ is sufficiently smooth. Rates are discussed 
in Section~\ref{subsec:rates}.

\paragraph{(ii) Approximation error to $\pi_0(\cdot\,| \, X)$.}
The factor $\pi_{\mathcal{P}(X)^\perp}(\cdot\,| \, X)$ is the component of 
$\pi_0(\cdot\,| \, X)$ orthogonal to the prior-weighted likelihood span $\mathcal{P}(X)$. It vanishes when $\pi_0(\cdot\,| \, x)\in\mathcal{P}(x)$ for all $x\in\cal X$. In that case, the AOI estimator exactly recovers $\pi_0$ by inverting the
outcome probabilities $\mathbb{P}(Y=y\,|\,X=x)$, and the bias is zero.
The residual norm
$\|\pi_{\mathcal{P}(X)^\perp}\|_{L^2(\mathcal{A})}$ will converge to zero when $\pi_0(\cdot\,| \, X)$ is regular enough. Rates are discussed 
in Section~\ref{subsec:rates}. 

\paragraph{(iii) Rate double robustness.}
Since the bias is the expectation of the inner product of two 
residuals in $L^2(\cal A)$, it is governed by residual product. By the Cauchy--Schwarz inequality,
\begin{equation}\label{eq:CS_bound}
\left|\mu_*^{(\infty)}-\mu_0\right|
\le
\mathbb{E}\!\left[\left\|\mu_{\mathcal{F}(X)^\perp}(X,\cdot)
\right\|_{L^2(\mathcal{A})}
\left\|\pi_{\mathcal{P}(X)^\perp}(\cdot\,| \, X)
\right\|_{L^2(\mathcal{A})}\right].
\end{equation}
Thus, the bias vanishes exactly when either 
$\mu(x,\cdot)\in\mathcal{F}(x)$ or 
$\pi_0(\cdot\,| \, x)\in\mathcal{P}(x)$, for all $x\in\mathcal{X}$ and is small whenever either 
approximation is good. The product structure is analogous to the double 
robustness property in semiparametric estimation 
\citep{funk2011doubly,chernozhukov2018double}, but here it operates at the level 
of rates rather than point identification. The following corollary makes 
this explicit.

\begin{corollary}\label{cor:rates}
  Suppose that, as $T\to\infty$,
  \begin{align*}
    \left\|\mu_{\mathcal{F}(X)^\perp}(X,\cdot)
    \right\|_{L^2(\mathcal{A})} 
    &= O_P(r_{\mu,T}), \\
    \left\|\pi_{\mathcal{P}(X)^\perp}(\cdot\,| \, X)
    \right\|_{L^2(\mathcal{A})}
    &= O_P(r_{\pi,T}).
  \end{align*}
  Then 
  $|\mu_*^{(\infty)}-\mu_0|=O(r_{\mu,T}\,r_{\pi,T})$.
\end{corollary}

\noindent
The rate $r_{\mu,T}$ depends on the smoothness of $\mu(\cdot, \cdot)$, and $r_{\pi,T}$ 
on the smoothness of $\pi_0(\cdot\,|\,\cdot)$. Smoothness of either one suffices for 
consistency; smoothness of both yields faster rates.

\subsubsection{Rates of convergence of the approximation errors}
\label{subsec:rates}

We now derive the approximation rates $r_{\mu,T}$ and $r_{\pi,T}$. To keep 
the exposition concrete, we consider the following setting.

\begin{example}[\textbf{Static binary choice without covariates}]
   \label{example:NoCovStaticLogit}
   Consider a static binary choice model without covariates. The outcome $Y$ is 
   the number of successes out of $T$ trials, 
   ${\cal Y} = \{0,\dots,T\}$, and the fixed effect is scalar, 
   ${\cal A}\subset\mathbb{R}$. For a known strictly increasing link 
   function $p:\mathcal{A}\to(0,1)$,  the conditional outcome probabilities are
   \[
    f(y\,| \,\alpha) = \binom{T}{y} p(\alpha)^y(1-p(\alpha))^{T-y},
    \qquad y\in{\cal Y},\ \alpha\in{\cal A}.
   \]
   The choice $p(\alpha)=(1+e^{-\alpha})^{-1}$ gives the logit 
   model. Since there are no covariates, we suppress $x$ from all notation.
\end{example}

This example isolates the core approximation-theoretic challenge: for each unit $i$, only a single draw $Y_i\sim\mathrm{Bin}(T,p(A_i))$ is available to learn about $A_i$.

\paragraph{(i) Approximation rate to $\pi_0(\cdot)$.}

In Example~\ref{example:NoCovStaticLogit}, the space $\mathcal{P}$ 
consists of the functions of the form
\[
\pi_\nu(\alpha) = \pi_{\rm prior}(\alpha)
\sum_{y=0}^T \nu(y)\binom{T}{y}a^{y}(1-a)^{T-y}, 
\qquad a := p(\alpha),
\]
where $\nu(\cdot)$ is any function in $N:=\{\nu:{\cal Y}\to \mathbb{R}\}$.
The functions $\pi_\nu(\alpha)$ are prior-weighted linear combinations of Bernstein basis polynomials in $a=p(\alpha)$. Assume that ${\cal A}=[\underline{\alpha},\overline{\alpha}]$ is a 
compact interval, and let $[\underline{a},\overline{a}]=[p(\underline{\alpha}),p(\overline{\alpha})]$.
Bounding $\|\pi_0(\cdot)-\pi_\nu(\cdot)\|_{L^2(\mathcal{A})}$ then reduces to bounding 
\begin{equation}\label{eq:to_bound}
\inf_{\nu\in N}\sup_{a\in[\underline{a},\overline{a}]}
\left|\pi_0(p^{-1}(a))-    \pi_{\rm prior}(p^{-1}(a))
    \sum_{y=0}^T  \nu(y) \, \binom{T}{y}a^{ y}(1-a)^{T-y}\right|.
\end{equation}
Assuming also that $\pi_{\rm prior}(\cdot)$ is bounded away from zero on $\mathcal{A}$, 
\eqref{eq:to_bound} is bounded if 
\begin{equation}\label{eq:to_bound2}
    \inf_{\nu\in N}\sup_{a\in[\underline{a},\overline{a}]}
    \left|\frac{\pi_0(p^{-1}(a))}{\pi_{\rm prior}(p^{-1}(a))}-    
      \sum_{y=0}^T  \nu(y) \, \binom{T}{y}a^{ y}(1-a)^{T-y}\right|
\end{equation} 
is bounded.
Since the Bernstein basis polynomials $\binom{T}{y}a^{y}(1-a)^{T-y}, y\in\cal Y$, span the 
space of polynomials of degree at most $T$, we can invoke standard results from polynomial approximation 
theory to bound \eqref{eq:to_bound2}. The rate of 
convergence depends on the smoothness of the function 
$a\mapsto \pi_0(p^{-1}(a))/\pi_{\rm prior}(p^{-1}(a))$ on $[\underline{a},\overline{a}]$. 
In particular, by Theorem 8.1 in \cite{devore1993constructive}, we have the following 
lemma.

\begin{lemma}\label{lmm.approx}
   Suppose the data-generating process is that of 
   Example~\ref{example:NoCovStaticLogit}, with ${\cal A}=[\underline{\alpha},\overline{\alpha}]$ 
   compact, and let $\pi_{\rm prior}$ be bounded away from zero on $\mathcal{A}$. 
   If the mapping $a\mapsto\pi_0(p^{-1}(a))/\pi_{\rm prior}(p^{-1}(a))$ is analytic 
   on $[\underline{a},\overline{a}]=[p(\underline{\alpha}),p(\overline{\alpha})]$, then there exists $0<\rho_\pi<1$ such that 
   \[
   \inf_{\nu\in N}\sup_{\alpha\in{\cal A}}
   |\pi_0(\alpha)-\pi_\nu(\alpha)|=O(\rho_\pi^T),
   \]
   and, therefore, $\left\|\pi_{\mathcal{P}^\perp}(\cdot)\right\|_{L^2(\mathcal{A})}
    = O(r_{\pi,T})$ with $r_{\pi,T}=\rho_\pi^T$.
\end{lemma}

\begin{remark}[Relaxing analyticity]
\label{rmk:relaxing_pi}
If instead $a\mapsto\pi_0(p^{-1}(a))/\pi_{\rm prior}(p^{-1}(a))$ has $s$ 
continuous derivatives on $[\underline{a},\overline{a}]$, standard results 
\citep{devore1993constructive} yield $r_{\pi,T}=T^{-s}$ instead of  
$r_{\pi,T}=\rho_\pi^T$.
\end{remark}

\paragraph{(ii) Approximation rate to $\mu(\cdot)$.}

A similar approximation argument applies to $\mu(\cdot)$, yielding the following result.

\begin{lemma}\label{lmm.approxmu}
   Suppose the data-generating process is that of 
   Example~\ref{example:NoCovStaticLogit}, with 
   ${\cal A}=[\underline{\alpha},\overline{\alpha}]$ compact, and let $\pi_{\rm prior}$ be bounded away from zero on $\mathcal{A}$. If the mapping 
   $a\mapsto\mu(p^{-1}(a))/\pi_{\rm prior}(p^{-1}(a))$ is 
   analytic on $[\underline{a},\overline{a}]$, then there exists $0<\rho_\mu<1$ such 
   that 
   $\sup_{\alpha\in{\cal A}}|\mu_{\mathcal{F}^\perp}(\alpha)|=O(\rho_\mu^T)$. Therefore,
   $\|\mu_{\mathcal{F}^\perp}(\cdot)\|_{L^2(\mathcal{A})}=O(r_{\mu,T})$ with $r_{\mu,T}=\rho_\mu^T$.
\end{lemma}

\begin{remark}[Relaxing analyticity]
\label{rmk:relaxing_mu}
As in Remark~\ref{rmk:relaxing_pi}, the rate becomes $r_{\mu,T}=T^{-s}$ when the mapping 
$a\mapsto\mu(p^{-1}(a))/\pi_{\rm prior}(p^{-1}(a))$ has $s$ 
continuous derivatives.
\end{remark}

\paragraph{Summary.}
Combining Lemmas~\ref{lmm.approx} and~\ref{lmm.approxmu} with 
Corollary~\ref{cor:rates} yields the following conclusion. If both mappings are analytic, the bias decays at rate $O(\rho_\pi^T\rho_\mu^T)$, that is, exponentially in $T$. If only one 
mapping---say, the one involving $\pi_0$---is analytic while the other has $s$ derivatives, the bias decays at rate $O(\rho_\pi^T  T^{-s})$, which remains exponentially fast in $T$.
\subsection{Inference}
\label{sec:InferenceKnownTheta}

We establish asymptotic normality of $\widehat{\mu}^{(\infty)}$ under triangular array
asymptotics where $n \to \infty$ and $T = T_n \to \infty$. Define
$$
\sigma_T^2 := {\rm Var}\left(w^{(\infty)}(Y,X)\right).
$$
We assume that $\sigma^2_T < \infty$ for all $T$; this is implicit in 
Assumption~\ref{ass:lindeberg} below.

\begin{assumption}\label{ass:iid}
The random vectors $(Y_i, X_i, A_i)$, $i = 1, \ldots, n$, are independent and 
identically distributed for each $T$.
\end{assumption}

\begin{assumption}\label{ass:variance_lower_bound}
There exists $\underline{\sigma}^2 > 0$ such that $\sigma_T^2 \geq \underline{\sigma}^2$ 
for all $T$.
\end{assumption}

\begin{assumption}\label{ass:lindeberg}
For all $\epsilon > 0$,
$$
\mathbb{E}\left[\left(\frac{w^{(\infty)}(Y,X) - \mu_*^{(\infty)}}{\sigma_T}\right)^2 
\mathbf{1}\left\{\left|\frac{w^{(\infty)}(Y,X) - \mu_*^{(\infty)}}{\sigma_T}\right| 
> \epsilon\sqrt{n}\right\}\right] \to 0
$$
as $n,T\to \infty$.
\end{assumption}

Assumption~\ref{ass:iid} restates the sampling assumption from 
Section~\ref{sec:Setup} for clarity. Assumption~\ref{ass:variance_lower_bound} is 
a nondegeneracy condition ensuring that 
sampling variability of $w^{(\infty)}(Y,X)$ does not vanish as $T\to\infty$.
Assumption~\ref{ass:lindeberg} is the Lindeberg condition adapted to the 
triangular array setting; it controls the tail behavior of the standardized influence 
function uniformly across the sequence $(n, T_n)$.

\begin{theorem}\label{thm:asymptotic_normality}
Suppose that Assumptions~\ref{ass:iid}--\ref{ass:lindeberg} hold, and that $|\mu_*^{(\infty)}-\mu_0|=O(r_{T})$.
If $n,T \to \infty$ such that 
\begin{equation}\label{eq:rate_condition}
    \sqrt{n} \, r_T \to 0,
\end{equation}
then 
$$
\sqrt{n}\left(\frac{\widehat{\mu}^{(\infty)} - \mu_0}{\sigma_T}\right)
\xrightarrow{d} {\cal N}(0,1).
$$
\end{theorem}

When $r_T = \rho^T$ for some $0<\rho <1$, as discussed in 
Section~\ref{subsec:bias}, the condition $\sqrt{n} \, r_T \to 0$ becomes 
$\sqrt{n} \, \rho^T \to 0$ or, equivalently,
$$
T \gg \frac{\log n}{2|\log \rho|}.
$$
Thus, $T$ needs only grow slightly faster than logarithmically in $n$ for valid 
inference. This is a remarkably mild requirement, and stands in contrast with 
the $T \gg n^{\frac{1}{2(s+1)}}$ condition typically needed, for $s$-th order bias correction,
when the bias decays at a polynomial rate $O(T^{-(s+1)})$.

For feasible inference, the variance $\sigma_T^2$ must be estimated. Define
\begin{equation}\label{eq:variance_estimator}
    \widehat{\sigma}_T^2 := \frac{1}{n} \sum_{i=1}^n 
    \left(w^{(\infty)}(Y_i,X_i) - \widehat{\mu}^{(\infty)}\right)^2.
\end{equation}
Consistency of $\widehat{\sigma}_T^2$ in the triangular array setting requires 
control on higher moments.

\begin{assumption}\label{ass:fourth_moment}
There exists $\kappa < \infty$ such that for all $T$,
$$
\mathbb{E}\left[\left(\frac{w^{(\infty)}(Y,X) - \mu_*^{(\infty)}}{\sigma_T}
\right)^4\right] \leq \kappa.
$$
\end{assumption}

\begin{corollary}[Feasible inference]\label{cor:feasible}
Under the conditions of Theorem~\ref{thm:asymptotic_normality} and 
Assumption~\ref{ass:fourth_moment}, we have 
$\widehat{\sigma}_T^2 \xrightarrow{p} \sigma_T^2$ and
$$
\sqrt{n}\left(\frac{\widehat{\mu}^{(\infty)} - \mu_0}{\widehat{\sigma}_T} \right)
\xrightarrow{d} {\cal N}(0,1).
$$
\end{corollary}

\begin{remark}[On Assumption~\ref{ass:variance_lower_bound}]
\label{rmk:variance_lower_bound}
Assumption~\ref{ass:variance_lower_bound} requires that the variance $\sigma_T^2$ 
remains bounded away from zero as $T\to\infty$. A natural lower bound on $\sigma_T^2$  arises from the 
law of total variance:
$$
\sigma_T^2 = {\rm Var}\left(\mathbb{E}\left[w^{(\infty)}(Y,X) 
\,\big|\, X\right]\right) 
+ \mathbb{E}\left[{\rm Var}\left(w^{(\infty)}(Y,X) \,\big|\, X\right)\right].
$$
Under the regularity conditions of Section~\ref{subsec:rates}, as $T \to \infty$,
$$
\mathbb{E}\left[w^{(\infty)}(Y,X) \,\big|\, X\right] \to 
\mathbb{E}\left[\mu(X,A) \,\big|\, X\right].
$$
Consequently, for large $T$,
$$
\sigma_T^2 \geq {\rm Var}\left(\mathbb{E}\left[w^{(\infty)}(Y,X) 
\,\big|\, X\right]\right) 
= {\rm Var}\left(\mathbb{E}\left[\mu(X,A) \,\big|\, X\right]\right) + o(1).
$$
This lower bound is strictly positive whenever the conditional average effect 
$\mathbb{E}[\mu(X,A)|X]$ varies with $X$, which holds generically.
\end{remark}

\section{Random-coefficient binary logit model: numerical results}\label{sec:RCnumerics}

We illustrate the AOI estimator numerically in the random-coefficient binary choice model of Section~\ref{subsec:regimeC-example}. Section~\ref{subsec:RCnumerics-twoblock} reports exact ($n=\infty$) bias and asymptotic standard deviation in the two-block design across a range of $T$ and iteration depths $q$. Section~\ref{subsec:RCnumerics-continuous} complements this with a Monte Carlo study at $n=500$ in a design with a continuous covariate, where we also assess coverage of the variance estimator.

\subsection{Two-block covariate design}\label{subsec:RCnumerics-twoblock}

We set $F=\Lambda$, where $\Lambda$ denotes the standard logistic distribution function, and retain the two-block design
\[
X_{it}=0 \quad \text{for } t\leq T/2,
\qquad
X_{it}=c_T \quad \text{for } t>T/2,
\]
with $T$ even. The average effect of interest is
$
\mu_0 = \mathbb{E}\big[\Lambda(A_{i,1}+A_{i,2})\big].
$
We set the true distribution of $(A_{i,1},A_{i,2})$ equal to the product of two independent logistic distributions with mean $1$ and scale $\frac{1}{2}$ (so the variance is $\frac{1}{4}\pi^2/3$). The prior used by AOI is deliberately severely misspecified: we set it equal to the product of two independent Gaussian distributions with mean $0$ and variance $4\pi^2/3$. Since the two-dimensional integrals entering $Q$ are not available in closed form, both the true distribution, $\pi_0$, and the prior, $\pi_{\rm prior}$, are approximated numerically by $K=1000$ quantiles in each dimension, each with mass $1/K$. The calculations are carried out in Matlab for $n=\infty$, $T\in\{2,4,6,8,10,20,30\}$, and $q\in\{0,1,2,5,10,50,100,1000,10^6,\infty\}$, and we also report a regularized version obtained by truncating eigenvalues of $Q$ below $\lambda_{\min}=10^{-4}$. Because the estimators are linear in the outcome frequencies, the reported bias is the exact fixed-$n$ bias, and the reported asymptotic standard deviation is the exact fixed-$n$ standard deviation times $\sqrt{n}$.

To keep the presentation compact, Tables~\ref{tab:1}--\ref{tab:3} report only $T\in\{2,6,20,30\}$ and $q\in\{0,1,2,10,1000,10^6,\infty\}$, together with the regularized AOI estimator. Here $\mathrm{AOI}(q)$ denotes the $q$-th order bias-corrected estimator, while AOI denotes the limit estimator $\mathrm{AOI}(\infty)$. We consider three scenarios of increasing difficulty to estimate $\mu_0$, whose true value, given our choice of $\pi_0$, is $0.8276$.

Scenario~1 sets $c_T=1$, so the target covariate value $x=1$ is directly observed in the second block. This is the point-identified benchmark in Section~\ref{subsec:regimeC-example}. Table~\ref{tab:1} shows that $\mathrm{AOI}(q)$ converges rapidly to AOI as $q$ increases, and that AOI has zero bias. By $q=10$, the bias is already numerically negligible for all four reported values of $T$, while the increase in asymptotic standard deviation relative to $q=0$ remains moderate. Up to $T=20$, regularization has essentially no visible effect. At $T=30$, however, regularization reduces the reported asymptotic standard deviation of AOI while leaving the bias essentially unchanged.

\begin{table}[htbp]
\centering
\caption{Scenario 1 ($c_T=1$): point-identified benchmark}
\label{tab:1}
\small
\setlength{\tabcolsep}{5pt}
\renewcommand{\arraystretch}{1.12}
\begin{tabular}{l rrrr@{\hspace{0.45cm}}rrrr}
\toprule
& \multicolumn{4}{c}{Bias} & \multicolumn{4}{c}{Asymptotic s.d.} \\
\cmidrule(lr){2-5}\cmidrule(lr){6-9}
& \multicolumn{1}{c}{$T=2$}
& \multicolumn{1}{c}{$T=6$}
& \multicolumn{1}{c}{$T=20$}
& \multicolumn{1}{c}{$T=30$}
& \multicolumn{1}{c}{$T=2$}
& \multicolumn{1}{c}{$T=6$}
& \multicolumn{1}{c}{$T=20$}
& \multicolumn{1}{c}{$T=30$} \\
\midrule
\multicolumn{9}{c}{\textit{Unregularized}} \\
$q=0$
& $-0.0864$ & $-0.0351$ & $-0.0105$ & $-0.0068$
& $0.2595$  & $0.2190$  & $0.1859$  & $0.1882$ \\
$q=1$
& $-0.0258$ & $-0.0044$ & $0.0001$  & $0.0003$
& $0.3300$  & $0.2462$  & $0.1938$  & $0.2263$ \\
$q=10$
& $-0.0001$ & $0.0002$  & $0.0001$  & $0.0001$
& $0.3774$  & $0.2548$  & $0.1950$  & $0.3468$ \\
$q=10^3$
& $0.0000$  & $-0.0000$ & $0.0000$  & $-0.0000$
& $0.3777$  & $0.2552$  & $0.1949$  & $0.2157$ \\
$q=10^6$
& $0.0000$  & $0.0000$  & $0.0000$  & $-0.0000$
& $0.3777$  & $0.2552$  & $0.1949$  & $1.1095$ \\
$q=\infty$
& $0.0000$  & $0.0000$  & $-0.0000$ & $0.0000$
& $0.3777$  & $0.2552$  & $0.1949$  & $0.2444$ \\
\midrule
\multicolumn{9}{c}{\textit{Regularized}} \\
$q=0$
& $-0.0864$ & $-0.0351$ & $-0.0105$ & $-0.0068$
& $0.2595$  & $0.2190$  & $0.1859$  & $0.1792$ \\
$q=1$
& $-0.0258$ & $-0.0044$ & $0.0001$  & $0.0003$
& $0.3300$  & $0.2462$  & $0.1938$  & $0.1843$ \\
$q=10$
& $-0.0001$ & $0.0002$  & $0.0001$  & $0.0001$
& $0.3774$  & $0.2548$  & $0.1950$  & $0.1847$ \\
$q=10^3$
& $0.0000$  & $-0.0000$ & $0.0000$  & $-0.0000$
& $0.3777$  & $0.2552$  & $0.1949$  & $0.1847$ \\
$q=10^6$
& $0.0000$  & $-0.0000$ & $0.0000$  & $-0.0000$
& $0.3777$  & $0.2552$  & $0.1949$  & $0.1847$ \\
$q=\infty$
& $0.0000$  & $-0.0000$ & $0.0000$  & $-0.0000$
& $0.3777$  & $0.2552$  & $0.1949$  & $0.1847$ \\
\bottomrule
\end{tabular}
\end{table}

Scenario~2 sets $c_T=1/2$. Here the target value $x=1$ is not observed, so estimation of $\mu_0$ requires extrapolation. Relative to $q=0$ and $q=1$, increasing $q$ reduces the bias substantially, but the asymptotic standard deviation can increase sharply. At $T=20$, the bias falls from $-0.077086$ at $q=0$ to $-4.100\times 10^{-4}$ for AOI, while the asymptotic standard deviation rises from $0.207323$ to $296.6$. At $T=30$, the same pattern is even more pronounced. Regularization leaves the bias essentially unchanged but dramatically reduces the asymptotic standard deviation of the AOI estimator.

\begin{table}[htbp]
\centering
\caption{Scenario 2 ($c_T=1/2$): extrapolation design}
\label{tab:2}
\small
\setlength{\tabcolsep}{5pt}
\renewcommand{\arraystretch}{1.12}
\begin{tabular}{l rrrr@{\hspace{0.45cm}}rrrr}
\toprule
& \multicolumn{4}{c}{Bias} & \multicolumn{4}{c}{Asymptotic s.d.} \\
\cmidrule(lr){2-5}\cmidrule(lr){6-9}
& \multicolumn{1}{c}{$T=2$}
& \multicolumn{1}{c}{$T=6$}
& \multicolumn{1}{c}{$T=20$}
& \multicolumn{1}{c}{$T=30$}
& \multicolumn{1}{c}{$T=2$}
& \multicolumn{1}{c}{$T=6$}
& \multicolumn{1}{c}{$T=20$}
& \multicolumn{1}{c}{$T=30$} \\
\midrule
\multicolumn{9}{c}{\textit{Unregularized}} \\
$q=0$
& $-0.1491$ & $-0.1103$ & $-0.0771$ & $-0.0640$
& $0.2417$  & $0.2253$  & $0.2073$  & $0.9672$ \\
$q=1$
& $-0.0978$ & $-0.0824$ & $-0.0533$ & $-0.0386$
& $0.3326$  & $0.2973$  & $0.2515$  & $0.6591$ \\
$q=10$
& $-0.0475$ & $-0.0663$ & $-0.0093$ & $-0.0013$
& $0.5884$  & $0.4536$  & $0.3154$  & $0.3790$ \\
$q=10^3$
& $-0.0414$ & $-0.0122$ & $0.0012$  & $0.0004$
& $0.6356$  & $0.8684$  & $0.6761$  & $3.9431$ \\
$q=10^6$
& $-0.0414$ & $-0.0120$ & $-0.0002$ & $-0.0001$
& $0.6356$  & $0.8750$  & $8.6447$  & $8.1778$ \\
$q=\infty$
& $-0.0414$ & $-0.0120$ & $-0.0004$ & $0.0000$
& $0.6356$  & $0.8750$  & $296.578$ & $37216.5$ \\
\midrule
\multicolumn{9}{c}{\textit{Regularized}} \\
$q=0$
& $-0.1491$ & $-0.1103$ & $-0.0771$ & $-0.0640$
& $0.2417$  & $0.2253$  & $0.2073$  & $0.1996$ \\
$q=1$
& $-0.0978$ & $-0.0824$ & $-0.0533$ & $-0.0386$
& $0.3326$  & $0.2973$  & $0.2515$  & $0.2335$ \\
$q=10$
& $-0.0475$ & $-0.0663$ & $-0.0093$ & $-0.0013$
& $0.5884$  & $0.4536$  & $0.3154$  & $0.2726$ \\
$q=10^3$
& $-0.0414$ & $-0.0122$ & $0.0012$  & $0.0004$
& $0.6356$  & $0.8684$  & $0.6757$  & $0.5109$ \\
$q=10^6$
& $-0.0414$ & $-0.0120$ & $0.0004$  & $-0.0001$
& $0.6356$  & $0.8750$  & $1.7061$  & $1.3399$ \\
$q=\infty$
& $-0.0414$ & $-0.0120$ & $0.0004$  & $-0.0001$
& $0.6356$  & $0.8750$  & $1.7061$  & $1.3399$ \\
\bottomrule
\end{tabular}
\end{table}
Scenario~3 sets $c_T=1/\sqrt{T}$, which is the weak-variation design highlighted in Section~\ref{subsec:regimeC-example}. This is the conceptually most interesting and challenging case: the slope $A_{i,2}$ is not consistently estimable, but the target average effect $\mu_0$ remains estimable. Table~\ref{tab:3} shows that the AOI sequence continues to reduce the bias toward zero, but the variance explosion is much stronger than in Scenario~2. At $T=20$, AOI has bias $1.1\times 10^{-3}$ and asymptotic standard deviation $4881$; at $T=30$, the reported asymptotic standard deviation of AOI is $1.433\times 10^{8}$. Regularization again stabilizes the computation sharply, reducing the asymptotic standard deviation at $T=30$ to $3.836$, while preserving the qualitative message that AOI still targets the average effect $\mu_0$ in this weak-variation design.

\begin{table}[htbp]
\centering
\caption{Scenario 3 ($c_T=1/\sqrt{T}$): weak-variation design}
\label{tab:3}
\small
\setlength{\tabcolsep}{5pt}
\renewcommand{\arraystretch}{1.12}
\begin{tabular}{l rrrr@{\hspace{0.45cm}}rrrr}
\toprule
& \multicolumn{4}{c}{Bias} & \multicolumn{4}{c}{Asymptotic s.d.} \\
\cmidrule(lr){2-5}\cmidrule(lr){6-9}
& \multicolumn{1}{c}{$T=2$}
& \multicolumn{1}{c}{$T=6$}
& \multicolumn{1}{c}{$T=20$}
& \multicolumn{1}{c}{$T=30$}
& \multicolumn{1}{c}{$T=2$}
& \multicolumn{1}{c}{$T=6$}
& \multicolumn{1}{c}{$T=20$}
& \multicolumn{1}{c}{$T=30$} \\
\midrule
\multicolumn{9}{c}{\textit{Unregularized}} \\
$q=0$
& $-0.1187$ & $-0.1299$ & $-0.1586$ & $-0.1660$
& $0.2552$  & $0.2194$  & $0.1987$  & $155.34$ \\
$q=1$
& $-0.0609$ & $-0.1048$ & $-0.1429$ & $-0.1487$
& $0.3410$  & $0.2996$  & $0.2886$  & $29.59$ \\
$q=10$
& $-0.0217$ & $-0.0871$ & $-0.0881$ & $-0.0893$
& $0.4718$  & $0.5144$  & $0.5349$  & $377.39$ \\
$q=10^3$
& $-0.0206$ & $-0.0234$ & $-0.0215$ & $-0.0228$
& $0.4778$  & $1.2317$  & $1.5608$  & $1112.43$ \\
$q=10^6$
& $-0.0206$ & $-0.0226$ & $-0.0015$ & $-0.0015$
& $0.4778$  & $1.2979$  & $24.3406$ & $1453.17$ \\
$q=\infty$
& $-0.0206$ & $-0.0226$ & $0.0011$  & $-0.0006$
& $0.4778$  & $1.2979$  & $4881.38$ & $1.4331\times 10^8$ \\
\midrule
\multicolumn{9}{c}{\textit{Regularized}} \\
$q=0$
& $-0.1187$ & $-0.1299$ & $-0.1586$ & $-0.1660$
& $0.2552$  & $0.2194$  & $0.1987$  & $0.1953$ \\
$q=1$
& $-0.0609$ & $-0.1048$ & $-0.1429$ & $-0.1487$
& $0.3410$  & $0.2996$  & $0.2886$  & $0.2881$ \\
$q=10$
& $-0.0217$ & $-0.0871$ & $-0.0881$ & $-0.0893$
& $0.4718$  & $0.5144$  & $0.5349$  & $0.5358$ \\
$q=10^3$
& $-0.0206$ & $-0.0234$ & $-0.0215$ & $-0.0228$
& $0.4778$  & $1.2317$  & $1.5597$  & $1.4751$ \\
$q=10^6$
& $-0.0206$ & $-0.0226$ & $-0.0097$ & $-0.0092$
& $0.4778$  & $1.2979$  & $3.8038$  & $3.8358$ \\
$q=\infty$
& $-0.0206$ & $-0.0226$ & $-0.0097$ & $-0.0092$
& $0.4778$  & $1.2979$  & $3.8038$  & $3.8358$ \\
\bottomrule
\end{tabular}
\end{table}
Taken together, the three tables make three points. First, when the target covariate value is observed, AOI reproduces the fixed-$T$ identified benchmark. Second, when estimation requires extrapolation, higher-order bias correction can remove most of the bias, though at a potentially large variance cost. Third, in the weak-variation design, AOI still reduces the bias toward zero, which is a key conceptual advantage of the method, but regularization becomes practically important because the high-order calculations are numerically unstable.

\subsection{Continuous covariate design}\label{subsec:RCnumerics-continuous}

We now turn to a Monte Carlo study with finite $n$ and a continuous covariate. The data-generating process is a binary choice logit model with unobserved heterogeneity and random coefficients, given by
\begin{align}
	Y_{it} = 1 \{X_{it}\, A_{i,2} + A_{i,1} \geq \varepsilon_{it} \},
	\label{eq:dgp1}
\end{align}
where
\begin{align}
	A_{i,1} \stackrel{\mathrm{i.i.d.}}{\sim}  \mathcal N (0,1),
	\qquad
	A_{i,2} \stackrel{\mathrm{i.i.d.}}{\sim} \mathcal N (1,1),
	\qquad
	X_{it} \mid A_{i,1},A_{i,2} \stackrel{\mathrm{i.i.d.}}{\sim}  \mathcal N (A_{i,1}+A_{i,2} ,1),
	\label{eq:dgp2}
\end{align}
and
\begin{align}
	\varepsilon_{it} \stackrel{\mathrm{i.i.d.}}{\sim} \mathrm{Logit}(0,1).
	\label{eq:dgp3}
\end{align}
We consider three average effects of interest. The first one is the average partial effect
\begin{align}
	\mathbb{E} \left[ \frac{\partial P(Y_{it}=1|X_{it},A_{i,1},A_{i,2})}{ \partial X_{it}} \right],
	\label{eq:ate}
\end{align}
which is a standard object of interest for a continuous covariate $X_{it}$. In addition to this, we also consider the average effects given by 
\begin{align}
	\mathbb{E} \left[ P(Y_{it}=1 | X_{it}=1, A_{i,1},A_{i,2}) \right],
	\label{eq:ae1}
\end{align}
and
\begin{align}
	\mathbb{E} \left[ P(Y_{it}=1 | X_{it}=0, A_{i,1},A_{i,2}) \right].
	\label{eq:ae2}
\end{align}
These correspond to $\mathbb{E}[\Lambda(A_{i,1}+A_{i,2})]$ and $\mathbb{E}[\Lambda(A_{i,1})]$, respectively.

We focus on panels with $N=500$ and $T\in\{2,4,6\}$. This choice of $T$ corresponds to our setting of interest with moderately large panels. All results are based on $K=1000$ replications. We consider three choices of prior functions, given by
	\begin{align*}
		\textbf{Prior 1} &: \pi_{\mathrm{prior}}(A_{i,1}|X_{i})\sim \mathrm{Logit}(1,1) \text{ and } \pi_{\mathrm{prior}}(A_{i,2}|X_{i})\sim \mathrm{Logit}(0,1),
		\\
		\textbf{Prior 2} &: \pi_{\mathrm{prior}}(A_{i,1}|X_{i})\sim \mathrm{Logit}(1,2) \text{ and } \pi_{\mathrm{prior}}(A_{i,2}|X_{i})\sim \mathrm{Logit}(0,2),
		\\
		\textbf{Prior 3} &: \pi_{\mathrm{prior}}(A_{i,1}|X_{i})\sim \mathcal N (0,1) \text{ and } \pi_{\mathrm{prior}}(A_{i,2}|X_{i})\sim \mathcal N (1,1), 
	\end{align*}
	which are all misspecified relative to the true distributions of $A_{i,1}$ and $A_{i,2}$. Note that Prior 3 is misspecified as it ignores the dependence between $(A_{i,1},A_{i,2})$ and $X_i$.
	 
	We calculate the AOI estimator $w^{(q)}(y,x)$ by numerical integration over a discretised support for $(A_{i,1},A_{i,2})$. In particular, for a given choice of priors, we obtain a grid of $L=99$ points for each  of $A_{i,1}$ and $A_{i,2}$, corresponding to equi-distant percentiles on the prior distributions of $A_{i,1}$ and $A_{i,2}$. This yields a grid size of $99 \times 99$. We also regularise $Q$ for numerical stability by clamping all eigenvalues smaller than $10^{-4}$ to $10^{-4}.$

The simulation results for the estimation of the average effects in \eqref{eq:ate}, \eqref{eq:ae1} and \eqref{eq:ae2} are presented in Tables \ref{tab:ae1}, \ref{tab:ae2} and \ref{tab:ae3}, respectively. Each table presents the average bias across replications, the standard deviation of estimators across replications, the ratio of estimated standard errors to simulation standard deviations (SE/SD), and the 95\% coverage rate of the AOI estimator. Several important patterns stand out. 

\afterpage{%
\clearpage
\begin{landscape}
\begin{table}[p]
\centering
\small
\begin{tabular}{lrrrrrr@{\hspace{1cm}}rrrrrr}
$T$ & \multicolumn{6}{c}{$q$} & \multicolumn{6}{c}{$q$} \\
\cmidrule(lr){2-7}\cmidrule(lr){8-13}
  & 0 & 100 & $10^3$ & $10^4$ & $10^5$ & $\infty$ & 0 & 100 & $10^3$ & $10^4$ & $10^5$ & $\infty$ \\
\midrule
\multicolumn{13}{l}{\textbf{Prior:} $\pi_{A_1} \sim \text{Logit}(1,1)$ and $\pi_{A_2} \sim \text{Logit}(0,1)$} \\
\midrule
 & \multicolumn{6}{l}{\quad \textbf{Bias}} & \multicolumn{6}{l}{\quad \textbf{SE/SD ratio}} \\
2 & -0.046 &  0.001 &  0.005 &  0.005 &  0.005 &  0.005 &  0.986 &  0.936 &  0.926 &  0.922 &  0.915 &  0.915 \\
4 & -0.029 &  0.009 &  0.009 &  0.009 &  0.008 &  0.008 &  1.013 &  0.984 &  0.980 &  0.974 &  0.970 &  0.970 \\
6 & -0.020 &  0.009 &  0.009 &  0.009 &  0.009 &  0.009 &  1.003 &  0.995 &  0.997 &  0.999 &  1.000 &  1.000 \\
\addlinespace[4pt]
 & \multicolumn{6}{l}{\quad \textbf{Standard deviation}} & \multicolumn{6}{l}{\quad \textbf{95\% coverage rate}} \\
2 &  0.003 &  0.032 &  0.064 &  0.112 &  0.132 &  0.132 &  0.000 &  0.933 &  0.930 &  0.938 &  0.939 &  0.939 \\
4 &  0.004 &  0.017 &  0.023 &  0.035 &  0.040 &  0.040 &  0.000 &  0.920 &  0.925 &  0.932 &  0.930 &  0.930 \\
6 &  0.004 &  0.012 &  0.015 &  0.021 &  0.023 &  0.023 &  0.003 &  0.883 &  0.908 &  0.927 &  0.936 &  0.936 \\
\midrule
\addlinespace[6pt]
\multicolumn{13}{l}{\textbf{Prior:} $\pi_{A_1} \sim \text{Logit}(1,2)$ and $\pi_{A_2} \sim \text{Logit}(0,2)$} \\
\midrule
 & \multicolumn{6}{l}{\quad \textbf{Bias}} & \multicolumn{6}{l}{\quad \textbf{SE/SD ratio}} \\
2 & -0.039 &  0.000 &  0.004 &  0.002 &  0.001 &  0.001 &  0.992 &  0.976 &  0.969 &  0.937 &  0.928 &  0.928 \\
4 & -0.016 &  0.010 &  0.008 &  0.008 &  0.007 &  0.007 &  0.999 &  1.002 &  0.997 &  0.971 &  0.971 &  0.971 \\
6 & -0.007 &  0.008 &  0.008 &  0.008 &  0.008 &  0.008 &  0.974 &  0.984 &  0.981 &  0.982 &  0.975 &  0.975 \\
\addlinespace[4pt]
 & \multicolumn{6}{l}{\quad \textbf{Standard deviation}} & \multicolumn{6}{l}{\quad \textbf{95\% coverage rate}} \\
2 &  0.005 &  0.045 &  0.085 &  0.152 &  0.179 &  0.179 &  0.000 &  0.941 &  0.945 &  0.953 &  0.956 &  0.956 \\
4 &  0.007 &  0.022 &  0.036 &  0.062 &  0.073 &  0.073 &  0.307 &  0.936 &  0.946 &  0.947 &  0.943 &  0.943 \\
6 &  0.007 &  0.015 &  0.023 &  0.038 &  0.044 &  0.044 &  0.811 &  0.915 &  0.931 &  0.942 &  0.944 &  0.944 \\
\midrule
\addlinespace[6pt]
\multicolumn{13}{l}{\textbf{Prior:} $\pi_{A_1} \sim \mathcal{N}(0,1)$ and $\pi_{A_2} \sim \mathcal{N}(1,1)$} \\
\midrule
 & \multicolumn{6}{l}{\quad \textbf{Bias}} & \multicolumn{6}{l}{\quad \textbf{SE/SD ratio}} \\
2 &  0.039 &  0.018 &  0.011 &  0.008 &  0.007 &  0.007 &  1.011 &  0.979 &  0.986 &  0.978 &  0.974 &  0.974 \\
4 &  0.037 &  0.011 &  0.008 &  0.008 &  0.008 &  0.008 &  0.996 &  1.006 &  1.003 &  0.988 &  0.974 &  0.974 \\
6 &  0.035 &  0.010 &  0.009 &  0.009 &  0.009 &  0.009 &  1.003 &  0.994 &  0.990 &  1.010 &  1.013 &  1.013 \\
\addlinespace[4pt]
 & \multicolumn{6}{l}{\quad \textbf{Standard deviation}} & \multicolumn{6}{l}{\quad \textbf{95\% coverage rate}} \\
2 &  0.003 &  0.023 &  0.046 &  0.083 &  0.097 &  0.097 &  0.000 &  0.868 &  0.925 &  0.937 &  0.943 &  0.943 \\
4 &  0.003 &  0.015 &  0.019 &  0.025 &  0.028 &  0.028 &  0.000 &  0.895 &  0.926 &  0.928 &  0.931 &  0.931 \\
6 &  0.003 &  0.011 &  0.013 &  0.015 &  0.016 &  0.016 &  0.000 &  0.854 &  0.895 &  0.916 &  0.924 &  0.924 \\
\bottomrule
\end{tabular}
\caption{Simulation results for $\mathbb{E} \left[ \partial P(Y_{it}=1|X_{it},A_{i,1},A_{i,2})/ \partial X_{it} \right]$ (true value: $0.066$) under the DGP given in \eqref{eq:dgp1}, \eqref{eq:dgp2} and \eqref{eq:dgp3}. Reported statistics include bias, standard deviation (SD), the ratio of estimated standard errors to simulation standard deviations (SE/SD), and 95\% coverage rates of the AOI estimator. All results are based on $K=1000$ replications and $N=500$.}
\label{tab:ae1}
\end{table}

\begin{table}[p]
\centering
\small
\begin{tabular}{lrrrrrr@{\hspace{1cm}}rrrrrr}
$T$ & \multicolumn{6}{c}{$q$} & \multicolumn{6}{c}{$q$} \\
\cmidrule(lr){2-7}\cmidrule(lr){8-13}
  & 0 & 100 & $10^3$ & $10^4$ & $10^5$ & $\infty$ & 0 & 100 & $10^3$ & $10^4$ & $10^5$ & $\infty$ \\
\midrule
\multicolumn{13}{l}{\textbf{Prior:} $\pi_{A_1} \sim \text{Logit}(1,1)$ and $\pi_{A_2} \sim \text{Logit}(0,1)$} \\
\midrule
 & \multicolumn{6}{l}{\quad \textbf{Bias}} & \multicolumn{6}{l}{\quad \textbf{SE/SD ratio}} \\
2 &  0.032 &  0.020 &  0.017 &  0.017 &  0.017 &  0.017 &  0.994 &  0.958 &  0.947 &  0.911 &  0.899 &  0.899 \\
4 &  0.039 &  0.012 &  0.008 &  0.004 &  0.002 &  0.002 &  1.001 &  0.997 &  0.998 &  0.998 &  1.002 &  1.002 \\
6 &  0.039 &  0.005 &  0.001 & -0.001 & -0.002 & -0.002 &  1.007 &  0.995 &  1.003 &  0.975 &  0.966 &  0.966 \\
\addlinespace[4pt]
 & \multicolumn{6}{l}{\quad \textbf{Standard deviation}} & \multicolumn{6}{l}{\quad \textbf{95\% coverage rate}} \\
2 &  0.008 &  0.031 &  0.054 &  0.094 &  0.109 &  0.109 &  0.032 &  0.874 &  0.910 &  0.936 &  0.943 &  0.943 \\
4 &  0.009 &  0.024 &  0.038 &  0.067 &  0.079 &  0.079 &  0.006 &  0.910 &  0.942 &  0.955 &  0.957 &  0.957 \\
6 &  0.009 &  0.021 &  0.033 &  0.059 &  0.071 &  0.071 &  0.013 &  0.932 &  0.956 &  0.950 &  0.943 &  0.943 \\
\midrule
\addlinespace[6pt]
\multicolumn{13}{l}{\textbf{Prior:} $\pi_{A_1} \sim \text{Logit}(1,2)$ and $\pi_{A_2} \sim \text{Logit}(0,2)$} \\
\midrule
 & \multicolumn{6}{l}{\quad \textbf{Bias}} & \multicolumn{6}{l}{\quad \textbf{SE/SD ratio}} \\
2 &  0.022 &  0.018 &  0.018 &  0.019 &  0.019 &  0.019 &  1.039 &  0.987 &  0.992 &  0.958 &  0.940 &  0.940 \\
4 &  0.030 &  0.009 &  0.005 &  0.005 &  0.005 &  0.005 &  1.008 &  0.976 &  0.993 &  1.016 &  1.017 &  1.017 \\
6 &  0.030 &  0.003 & -0.001 & -0.004 & -0.005 & -0.005 &  0.999 &  1.010 &  1.010 &  0.995 &  0.992 &  0.992 \\
\addlinespace[4pt]
 & \multicolumn{6}{l}{\quad \textbf{Standard deviation}} & \multicolumn{6}{l}{\quad \textbf{95\% coverage rate}} \\
2 &  0.011 &  0.032 &  0.052 &  0.089 &  0.105 &  0.105 &  0.514 &  0.902 &  0.935 &  0.953 &  0.956 &  0.956 \\
4 &  0.012 &  0.027 &  0.045 &  0.079 &  0.093 &  0.093 &  0.295 &  0.924 &  0.948 &  0.955 &  0.960 &  0.960 \\
6 &  0.012 &  0.024 &  0.041 &  0.073 &  0.088 &  0.088 &  0.263 &  0.944 &  0.944 &  0.956 &  0.956 &  0.956 \\
\midrule
\addlinespace[6pt]
\multicolumn{13}{l}{\textbf{Prior:} $\pi_{A_1} \sim \mathcal{N}(0,1)$ and $\pi_{A_2} \sim \mathcal{N}(1,1)$} \\
\midrule
 & \multicolumn{6}{l}{\quad \textbf{Bias}} & \multicolumn{6}{l}{\quad \textbf{SE/SD ratio}} \\
2 &  0.024 &  0.011 &  0.005 &  0.003 &  0.002 &  0.002 &  0.973 &  0.989 &  0.982 &  0.975 &  0.974 &  0.974 \\
4 &  0.028 &  0.001 & -0.002 & -0.002 & -0.002 & -0.002 &  0.996 &  0.980 &  0.970 &  0.966 &  0.969 &  0.969 \\
6 &  0.028 & -0.003 & -0.006 & -0.007 & -0.007 & -0.007 &  0.981 &  0.959 &  0.972 &  1.003 &  1.005 &  1.005 \\
\addlinespace[4pt]
 & \multicolumn{6}{l}{\quad \textbf{Standard deviation}} & \multicolumn{6}{l}{\quad \textbf{95\% coverage rate}} \\
2 &  0.005 &  0.026 &  0.047 &  0.079 &  0.092 &  0.092 &  0.000 &  0.923 &  0.941 &  0.948 &  0.951 &  0.951 \\
4 &  0.006 &  0.023 &  0.034 &  0.053 &  0.062 &  0.062 &  0.002 &  0.936 &  0.946 &  0.940 &  0.944 &  0.944 \\
6 &  0.006 &  0.020 &  0.028 &  0.042 &  0.049 &  0.049 &  0.006 &  0.946 &  0.952 &  0.958 &  0.961 &  0.961 \\
\bottomrule
\end{tabular}
\caption{Simulation results for $\mathbb{E}[P(Y_{it}=1|X_{it}=1, A_{i,1}, A_{i,2})]$ (true value: $0.679$). See Table \ref{tab:ae1} for more information.}
\label{tab:ae2}
\end{table}

\begin{table}[p]
\centering
\small
\begin{tabular}{lrrrrrr@{\hspace{1cm}}rrrrrr}
$T$ & \multicolumn{6}{c}{$q$} & \multicolumn{6}{c}{$q$} \\
\cmidrule(lr){2-7}\cmidrule(lr){8-13}
  & 0 & 100 & $10^3$ & $10^4$ & $10^5$ & $\infty$ & 0 & 100 & $10^3$ & $10^4$ & $10^5$ & $\infty$ \\
\midrule
\multicolumn{13}{l}{\textbf{Prior:} $\pi_{A_1} \sim \text{Logit}(1,1)$ and $\pi_{A_2} \sim \text{Logit}(0,1)$} \\
\midrule
 & \multicolumn{6}{l}{\quad \textbf{Bias}} & \multicolumn{6}{l}{\quad \textbf{SE/SD ratio}} \\
2 &  0.149 &  0.058 &  0.050 &  0.050 &  0.051 &  0.051 &  0.984 &  1.000 &  0.994 &  0.987 &  0.978 &  0.978 \\
4 &  0.129 &  0.024 &  0.018 &  0.015 &  0.014 &  0.014 &  1.001 &  1.001 &  0.971 &  0.978 &  0.984 &  0.984 \\
6 &  0.112 &  0.014 &  0.011 &  0.011 &  0.010 &  0.010 &  1.013 &  0.992 &  0.994 &  1.008 &  1.005 &  1.005 \\
\addlinespace[4pt]
 & \multicolumn{6}{l}{\quad \textbf{Standard deviation}} & \multicolumn{6}{l}{\quad \textbf{95\% coverage rate}} \\
2 &  0.007 &  0.030 &  0.053 &  0.089 &  0.104 &  0.104 &  0.000 &  0.529 &  0.834 &  0.900 &  0.915 &  0.915 \\
4 &  0.008 &  0.026 &  0.045 &  0.077 &  0.091 &  0.091 &  0.000 &  0.829 &  0.910 &  0.945 &  0.953 &  0.953 \\
6 &  0.008 &  0.024 &  0.040 &  0.067 &  0.079 &  0.079 &  0.000 &  0.894 &  0.930 &  0.943 &  0.943 &  0.943 \\
\midrule
\addlinespace[6pt]
\multicolumn{13}{l}{\textbf{Prior:} $\pi_{A_1} \sim \text{Logit}(1,2)$ and $\pi_{A_2} \sim \text{Logit}(0,2)$} \\
\midrule
 & \multicolumn{6}{l}{\quad \textbf{Bias}} & \multicolumn{6}{l}{\quad \textbf{SE/SD ratio}} \\
2 &  0.100 &  0.037 &  0.034 &  0.037 &  0.038 &  0.038 &  1.008 &  0.983 &  1.003 &  0.990 &  0.978 &  0.978 \\
4 &  0.072 &  0.009 &  0.009 &  0.012 &  0.012 &  0.012 &  0.998 &  0.964 &  0.967 &  0.944 &  0.943 &  0.943 \\
6 &  0.054 &  0.003 &  0.004 &  0.005 &  0.004 &  0.004 &  0.979 &  0.979 &  0.988 &  0.985 &  0.987 &  0.987 \\
\addlinespace[4pt]
 & \multicolumn{6}{l}{\quad \textbf{Standard deviation}} & \multicolumn{6}{l}{\quad \textbf{95\% coverage rate}} \\
2 &  0.010 &  0.035 &  0.057 &  0.097 &  0.113 &  0.113 &  0.000 &  0.790 &  0.916 &  0.940 &  0.942 &  0.942 \\
4 &  0.011 &  0.031 &  0.051 &  0.080 &  0.092 &  0.092 &  0.000 &  0.930 &  0.936 &  0.941 &  0.947 &  0.947 \\
6 &  0.011 &  0.028 &  0.047 &  0.073 &  0.085 &  0.085 &  0.002 &  0.947 &  0.951 &  0.942 &  0.946 &  0.946 \\
\midrule
\addlinespace[6pt]
\multicolumn{13}{l}{\textbf{Prior:} $\pi_{A_1} \sim \mathcal{N}(0,1)$ and $\pi_{A_2} \sim \mathcal{N}(1,1)$} \\
\midrule
 & \multicolumn{6}{l}{\quad \textbf{Bias}} & \multicolumn{6}{l}{\quad \textbf{SE/SD ratio}} \\
2 &  0.008 & -0.008 & -0.011 & -0.011 & -0.012 & -0.012 &  0.980 &  0.996 &  0.985 &  0.977 &  0.977 &  0.977 \\
4 &  0.008 & -0.015 & -0.014 & -0.015 & -0.015 & -0.015 &  0.985 &  0.970 &  0.978 &  1.020 &  1.028 &  1.028 \\
6 &  0.006 & -0.015 & -0.013 & -0.012 & -0.012 & -0.012 &  0.959 &  0.982 &  0.986 &  1.006 &  1.006 &  1.006 \\
\addlinespace[4pt]
 & \multicolumn{6}{l}{\quad \textbf{Standard deviation}} & \multicolumn{6}{l}{\quad \textbf{95\% coverage rate}} \\
2 &  0.004 &  0.025 &  0.047 &  0.083 &  0.096 &  0.096 &  0.366 &  0.937 &  0.945 &  0.943 &  0.946 &  0.946 \\
4 &  0.004 &  0.023 &  0.038 &  0.062 &  0.073 &  0.073 &  0.563 &  0.886 &  0.932 &  0.946 &  0.949 &  0.949 \\
6 &  0.005 &  0.021 &  0.033 &  0.053 &  0.062 &  0.062 &  0.766 &  0.894 &  0.938 &  0.953 &  0.956 &  0.956 \\
\bottomrule
\end{tabular}
\caption{Simulation results for $\mathbb{E}[P(Y_{it}=1|X_{it}=0, A_{i,1}, A_{i,2})]$ (true value: $0.500$). See Table \ref{tab:ae1} for more information.}
\label{tab:ae3}
\end{table}

\end{landscape}
\clearpage
}

Regarding bias, in all cases the AOI estimator reduces the bias substantially compared to $q=0$, especially at $q=\infty$. Interestingly, the bias of the average partial effect $\mathbb{E} \left[ \partial P(Y_{it}=1|X_{it},A_{i,1},A_{i,2})/ \partial X_{it} \right]$ remains just below 0.01 even for $q=\infty$; see Table \ref{tab:ae1}. However, it is possible that this average effect is inherently difficult to estimate and requires greater $T$ than considered here. For $\mathbb{E} \left[ P(Y_{it}=1 | X_{it}=1, A_{i,1}, A_{i,2}) \right]$  bias appears to be small relative to its true value of $0.679$, even at $q=0$; see Table \ref{tab:ae2}. The average effect $\mathbb{E} \left[ P(Y_{it}=1 | X_{it}=0, A_{i,1}, A_{i,2}) \right]$, on the other hand, shows an interesting pattern: depending on the choice of priors, the bias can be quite small or significantly large at $q=0$; see Table \ref{tab:ae3}. All in all, the results for the uncorrected case of $q=0$ show that the severity of bias also depends on the average effect itself. Nevertheless, at $q=\infty$ and $T=6$ the bias becomes negligibly small in most settings.

The simulation results also reveal that estimator variance increases with $q$. This is not an unexpected reflection of the classical bias-variance trade-off. However, while the estimator standard deviation has a clearly increasing trend with $q$, we do not observe an explosive behaviour. As for the estimation of the standard deviation, across most configurations the variance estimator is quite accurate (as revealed by the SE/SD ratios). 

Finally---and most importantly---in almost all cases, for $q=\infty$ (and even for many large but finite $q$ settings) the coverage rates are very close to the nominal coverage rate of 95\%. The coverage rates at $q=0$, on the other hand, are in stark contrast to this result and often fall below 0.5. This confirms the validity of the asymptotic distribution as $q\to\infty$ even for very small values of $T$ in this complicated model, and strongly supports the validity of our approach in terms of inference.

\section{Inference on common parameters}
\label{sec:theta}

The main body of this paper focuses on estimation of $\mu_0$ in models where
the conditional outcome probabilities $f(y\,|\,x,\alpha)$ do not depend on any
common parameter. As noted before,
many panel models of interest include a finite-dimensional common parameter
$\theta_0$, so that the outcome probabilities take the form
$f(y\,|\,x,\alpha,\theta_0)$. This section discusses two issues that arise in
that setting: how inference on $\mu_0$ is affected by the presence of
$\theta_0$, and how $\theta_0$ itself can be estimated.

\subsection{Inference on $\mu_0$ in models with $\theta_0$}
\label{subsec:mu-with-theta}

All results in this paper carry over directly to models with a common parameter
$\theta_0$, provided a consistent estimator $\widehat{\theta}$ of $\theta_0$ is
available. Given $\widehat{\theta}$, one simply evaluates the AOI estimating
function at $\widehat{\theta}$, that is,
$\widehat{\mu}^{(\infty)} = n^{-1} \sum_{i=1}^n
w^{(\infty)}(Y_i, X_i, \widehat{\theta})$, where $w^{(\infty)}(y,x,\theta)$
is the estimating function from Section~\ref{sec:estimator} applied to the
model $f(y\,|\,x,\alpha,\theta)$.

Consistent estimators of $\theta_0$ are available in a wide range of nonlinear
panel models. For essentially every type of discrete outcome variable (binary,
count data, ordered choice, multinomial choice), there exist model
specifications that allow point identification and $\sqrt{n}$-consistent
estimation of $\theta_0$ even at fixed $T$. In static models, this is typically
achieved through conditional likelihood methods that exploit the existence of a
sufficient statistic for $A_i$, as in exponential-family models
\citep{rasch1961general, andersen1970asymptotic, chamberlain1980analysis}.
In dynamic models, appropriate specifications similarly allow estimation of
$\theta_0$ via generalized method of moments
\citep[see, e.g.,][]{honore2020dynamic}. More generally, the functional
differencing method of \cite{bonhomme2012functional} provides a unifying
framework for point estimation of $\theta_0$ in both static and dynamic panel
models. These methods are well established in the literature and widely
implemented in statistical software. In models where $\theta_0$ is not
point-identified at fixed $T$, the approximate functional differencing (AFD)
method of \cite{dhaene2023approximate} can be used to obtain consistent
estimators of $\theta_0$ whose bias vanishes rapidly as $T$ grows.

When $\theta_0$ is estimated, the estimation error in $\widehat{\theta}$
may need to be accounted for when conducting inference on $\mu_0$. If
$\widehat{\theta}$ is $\sqrt{nT}$-consistent, as is the case for many
conditional likelihood and AFD estimators, then
$\sqrt{n}(\widehat{\theta} - \theta_0) = O_P(T^{-1/2})$, and the contribution
of the estimation error in $\widehat{\theta}$ to the asymptotic distribution of
$\widehat{\mu}^{(\infty)}$ vanishes as $T \to \infty$. In this case, no formal
adjustment is required, and the asymptotic theory of
Section~\ref{sec:asymptotics} applies directly. For good finite-sample
performance, however, it may still be advisable to account for the estimation
error in $\widehat{\theta}$.

More generally, if $\widehat{\theta}$ is $\sqrt{n}$-consistent and
asymptotically linear with influence function $\psi(Y,X,\theta_0)$, and if
$\theta \mapsto w^{(\infty)}(y,x,\theta)$ is sufficiently smooth, then the
delta method gives
$$
\sqrt{n}(\widehat{\mu}^{(\infty)} - \widetilde{\mu}^{(\infty)})
\approx \sqrt{n}\, G(\theta_0)^\top (\widehat{\theta} - \theta_0),
$$
where $G(\theta_0) := \mathbb{E}[\nabla_\theta
w^{(\infty)}(Y,X,\theta)]_{\theta=\theta_0}$ and
$\widetilde{\mu}^{(\infty)}$ denotes the infeasible estimator evaluated at the
true $\theta_0$. The adjusted asymptotic variance becomes
$$
\widetilde{\sigma}_T^2 = \sigma_T^2
+ 2\, G(\theta_0)^\top \mathrm{Cov}(\psi, w^{(\infty)})
+ G(\theta_0)^\top \mathrm{Var}(\psi)\, G(\theta_0),
$$
and the analogue of Theorem~\ref{thm:asymptotic_normality} holds with
$\sigma_T^2$ replaced by $\widetilde{\sigma}_T^2$.

\subsection{Estimation of $\theta_0$}
\label{subsec:theta-estimation}

{The ideas developed in this paper also apply to inference on $\theta_0$ itself.
For most of the paper, we have suppressed $\theta_0$ from the notation.
Reinstating $\theta_0$, the AOI construction produces estimating functions
$w(Y_i, X_i, \theta_0)$ that satisfy $\mu_0 \approx \mathbb{E}[w(Y_i, X_i, \theta_0)]$
for large $T$, with the bias decaying exponentially under our regularity
conditions. Estimation of $\theta_0$ calls for moment functions of the same
family but targeting zero rather than $\mu_0$, that is, functions
$m(Y_i, X_i, \theta_0)$ such that
$0 \approx \mathbb{E}[m(Y_i, X_i, \theta_0)]$, which can be combined into a
GMM estimator for $\theta_0$.}

{The construction of such moment functions follows the same logic as the
construction of $w$ in this paper, starting from the score of the MLE for
$\theta$ and applying the iterated bias correction machinery to
the resulting average-effect expression --- this construction is detailed,
without asymptotic theory, in \cite{dhaene2023approximate}. At
$\theta = \theta_0$, the resulting moment functions are themselves average
effects in the sense of the present paper, so the bias and variance results
of Section~\ref{sec:asymptotics} apply directly to
$\mathbb{E}[m(Y_i, X_i, \theta_0)]$ as $n, T \to \infty$.
Two ingredients are not addressed in that paper: identification of $\theta_0$
from the moment conditions, and well-behavedness of the Jacobian
$\nabla_\theta \mathbb{E}[m(Y_i, X_i, \theta)]\big|_{\theta = \theta_0}$.
Once these are established, standard GMM cross-sectional asymptotics combined
with the bias control developed here yields the asymptotic distribution of
the GMM estimator for $\theta_0$. Controlling the bias at $\theta_0$ is the
technically demanding step, and it is precisely what the AOI theory of our
paper provides.}

\section{Conclusion} \label{sec:conclusion}  

This paper develops the approximate operator inversion method (AOI) for estimating average
effects in nonlinear panel data models with fixed effects. The central idea is
to recast the estimation problem as an inversion of the distributional mapping
from the fixed-effect distribution to the outcome distribution. This mapping goes from an
infinite-dimensional space to a finite-dimensional space, so it cannot be inverted exactly,
but the approximation improves as $T$ grows because the outcome space becomes
richer. The resulting estimator can be understood as the limit of infinitely iterated large $T$ bias corrections.

Two properties of AOI are worth highlighting. First, the bias has a product
structure (rate double robustness), decaying at the product of the
approximation rates for the average-effect function and the fixed-effect
distribution. Under analyticity conditions, this yields exponential bias decay
in $T$, so that only $T \gg \log n$ is needed for valid inference. Second, the
method operates entirely at the distributional level and never estimates
individual fixed effects. This makes it applicable in what we call regime (C) in the introduction, where the
fixed effects cannot be consistently estimated, a setting not covered or discussed by any existing papers.

\section*{Acknowledgment}
This paper benefited from the use of generative AI tools to assist with language editing and \LaTeX{} formatting; all output was carefully reviewed by the authors. All substantive content, results, and any remaining errors are the authors' responsibility.

\bibliographystyle{apalike}
\bibliography{apebib}

\appendix

\section{Proofs}

This appendix collects all proofs. Results are presented in the order they appear 
in the main text.

\subsection{Derivation of the plug-in bias}
\label{app:plugin_bias}

By the law of iterated expectations and model assumption \eqref{model},
\begin{align*}
\mathbb{E}\bigl[w^{(0)}(Y,X)\bigr]
&= \mathbb{E}\!\left[\sum_{y\in\cal Y} w^{(0)}(y,X) 
   \int_{\cal A} f(y\,| \, X,\alpha)\,\pi_0(\alpha\,| \, X)\,\dd\alpha\right]\\
&= \mathbb{E}\!\left[\int_{\cal A} \sum_{y\in\cal Y} w^{(0)}(y,X) 
   f(y\,| \, X,\alpha)\,\pi_0(\alpha\,| \, X)\,\dd\alpha\right],
\end{align*}
where the interchange of sum and integral is justified by finiteness of $\mathcal{Y}$.
Subtracting
\[
\mu_0 = \mathbb{E}\!\left[\int_{\cal A} \mu(X,\alpha)\,
\pi_0(\alpha\,| \, X)\,\dd\alpha\right]
\]
yields
\begin{align*}
\mathbb{E}\bigl[w^{(0)}(Y,X) - \mu_0\bigr] 
= \mathbb{E}\!\left[\int_{\cal A}  
\Bigl(\sum_{y\in \cal Y} w^{(0)}(y,X)\, f(y\,| \, X,\alpha)
-\mu(X,\alpha)\Bigr)\,\pi_0(\alpha\,| \, X)\,\dd\alpha \right],
\end{align*}
as claimed in Section~\ref{sec:iterated}.

\subsection{Proofs of the results of Section~\ref{subsec:bias}}

\subsubsection{Notation} 
Let ${\cal Y}=\{y_{(1)},\ldots,y_{(n_{\cal Y})}\}$ and fix 
$x\in{\cal X}$. Define
\[
P_{\rm p}(x) := \diag\!\Big(p_{\rm prior}(y_{(k)}\,| \, x)
\Big)_{k=1,\ldots,n_{\cal Y}},
\]
and let $\pi_{\rm p}(x)$ denote the linear operator on $L^2(\mathcal{A})$ given, for all 
$v\in L^2(\mathcal{A})$, by
\[
\bigl(\pi_{\rm p}(x)\,v\bigr)(\alpha) \;=\; \pi_{\rm prior}(\alpha\,| \, x)\,v(\alpha).
\]
In what follows we also treat $\pi_0(\cdot\,| \, x)$ and $\mu(x,\cdot)$ as elements 
of $L^2(\mathcal{A})$.

For each $x\in{\cal X}$, define the operator
$\mathbf{F}(x):L^2(\mathcal{A})\to\mathbb{R}^{n_{\cal Y}}$
by
\[
\mathbf{F}(x)\,v
\;:=\;
\begin{pmatrix}
\displaystyle\int_{\mathcal{A}} f(y_{(1)}\,| \, x,\alpha)\,v(\alpha)\,\dd\alpha \\[6pt]
\vdots \\[6pt]
\displaystyle\int_{\mathcal{A}} f(y_{(n_{\cal Y})}\,| \, x,\alpha)\,v(\alpha)\,\dd\alpha
\end{pmatrix}.
\]
Note that $\mathbf{F}(x)$ is distinct from the likelihood function 
$f(\cdot\,| \,\cdot,\cdot)$ and from the subspace $\mathcal{F}(x)$ defined in 
Section~\ref{subsec:bias}. Next set
\[
G(x) \;:=\; \pi_{\rm p}(x)\,\mathbf{F}(x)^*\,P_{\rm p}(x)^{-1},
\]
where $\mathbf{F}(x)^*$ denotes the Hilbert adjoint of $\mathbf{F}(x)$. 
Then $G(x)$ is an operator $\mathbb{R}^{n_{\cal Y}}\to L^2(\mathcal{A})$ 
satisfying
\[
\bigl(G(x)\,v\bigr)(\alpha)
\;=\;\sum_{y\in{\cal Y}} \pi_{\rm post}(\alpha\,| \, y,x)\,v(y), 
\qquad v\in\mathbb{R}^{n_{\cal Y}}.
\]
It follows that
\[
Q(x) \;=\; \mathbf{F}(x)\,G(x).
\]
Finally, define the rescaled operators
\[
\widetilde{Q}(x) \;:=\; 
P_{\rm p}(x)^{-1/2}\,Q(x)\,P_{\rm p}(x)^{1/2},
\qquad
\widetilde{\mathbf{F}}(x) \;:=\; 
P_{\rm p}(x)^{-1/2}\,\mathbf{F}(x)\,\pi_{\rm p}(x)^{1/2}.
\]
Note that $\widetilde{Q}(x) = \widetilde{\mathbf{F}}(x)\,
\widetilde{\mathbf{F}}(x)^*$.
The rescaled operator $\widetilde{Q}$ is introduced because, unlike $Q$, it is real 
symmetric, which facilitates the spectral analysis below.

\subsubsection{Proof of Lemma~\ref{lemma.qinf}}

Fix $x\in{\cal X}$. As noted in the proof of Lemma 1 in 
\cite{dhaene2023approximate}, $\widetilde{Q}(x)$ is real symmetric and has 
eigenvalues in $[0,1]$. Hence there exist eigenvectors 
$\widetilde{u}_k(x)$, $k=1,\ldots,n_{\cal Y}$, and eigenvalues 
$\widetilde{\lambda}_k(x)\in[0,1]$, $k=1,\dots,n_{\cal Y}$, such that
\begin{align*}
\widetilde{Q}(x)
&=\sum_{k=1}^{n_{\cal Y}}\widetilde{\lambda}_k(x)\,  
\widetilde{u}_k(x)\,\widetilde{u}_k(x)^\top,\\[4pt]
\widetilde{Q}(x)^D
&=\sum_{k=1}^{n_{\cal Y}}
\frac{\mathbf{1}\{\widetilde{\lambda}_k(x)\ne 0\}}
{\widetilde{\lambda}_k(x)}\,
\widetilde{u}_k(x)\,\widetilde{u}_k(x)^\top.
\end{align*}
Write the partial sums of the Neumann series:
\[
S_q:=\sum_{r=0}^{q}\bigl(\mathbb{I}_{n_{\cal Y}}-\widetilde{Q}(x)\bigr)^r
=\sum_{k=1}^{n_{\cal Y}} \sum_{r=0}^{q}
\bigl(1-\widetilde{\lambda}_k(x)\bigr)^r\, 
\widetilde{u}_k(x)\,\widetilde{u}_k(x)^\top.
\]
For $0<\widetilde{\lambda}_k(x)\le 1$, we have 
$|1-\widetilde{\lambda}_k(x)|<1$ and therefore
\[
\sum_{r=0}^{\infty}\bigl(1-\widetilde{\lambda}_k(x)\bigr)^r=
\frac{1}{\widetilde{\lambda}_k(x)}.
\]
If $\widetilde{\lambda}_k(x)=0$, then 
$\sum_{r=0}^{q}(1-\widetilde{\lambda}_k(x))^r=q+1$ diverges as $q\to\infty$, 
but in this case 
$\widetilde{u}_k(x)\in\ker\bigl(\widetilde{Q}(x)\bigr)=
\operatorname{range}\!\bigl(\widetilde{\mathbf{F}}(x)\bigr)^{\perp}$,
so that $\widetilde{\mathbf{F}}(x)^{*}\widetilde{u}_k(x)=0$.
Hence,
\begin{align*}
\widetilde{\mathbf{F}}(x)^{*}S_q
&=\sum_{k=1}^{n_{\cal Y}} \sum_{r=0}^{q}
\mathbf{1}\{\widetilde{\lambda}_k(x)\ne 0\}\, 
\bigl(1-\widetilde{\lambda}_k(x)\bigr)^r\, 
\widetilde{\mathbf{F}}(x)^{*}\,\widetilde{u}_k(x)\,  
\widetilde{u}_k(x)^\top \\
&\;\xrightarrow{q\to\infty}\;
\sum_{k=1}^{n_{\cal Y}}
\frac{\mathbf{1}\{\widetilde{\lambda}_k(x)\ne 0\}}
{\widetilde{\lambda}_k(x)}\,
\widetilde{\mathbf{F}}(x)^{*}\,\widetilde{u}_k(x)\,  
\widetilde{u}_k(x)^\top 
\;=\; \widetilde{\mathbf{F}}(x)^{*}\,\widetilde{Q}(x)^{D},
\end{align*}
where convergence is in operator norm in 
$\mathcal{L}(\mathbb{R}^{n_{\mathcal{Y}}},L^2(\mathcal{A}))$.
By definition of $\widetilde{\mathbf{F}}(x)$ and $\widetilde{Q}(x)$, 
this yields 
\begin{align*}
&\pi_{\rm p}(x)\,\mathbf{F}(x)^*\,P_{\rm p}(x)^{-1}\,
Q(x)^D\,P_{\rm p}(x)^{1/2}\\
&\quad=\lim_{q\to\infty}\; \pi_{\rm p}(x)\,\mathbf{F}(x)^*\, 
P_{\rm p}(x)^{-1/2}\, 
\sum_{r=0}^{q}\Bigl(\mathbb{I}_{n_{\cal Y}}-
P_{\rm p}(x)^{-1/2}\,Q(x)\,P_{\rm p}(x)^{1/2}\Bigr)^r\\
&\quad=\lim_{q\to\infty}\; \pi_{\rm p}(x)\,\mathbf{F}(x)^*\, 
P_{\rm p}(x)^{-1}
\left[\sum_{r=0}^{q}\bigl(\mathbb{I}_{n_{\cal Y}}-Q(x)\bigr)^r\right] 
P_{\rm p}(x)^{1/2},
\end{align*}
which implies, using the definition of $G(x)$, that
\[
G(x)\,Q(x)^D
=\lim_{q\to\infty}\;G(x)
\sum_{r=0}^{q}\bigl(\mathbb{I}_{n_{\cal Y}}-Q(x)\bigr)^r.
\]
The vector of estimating functions satisfies, for each $k=1,\ldots,n_{\cal Y}$,
\[
w^{(\infty)}(y_{(k)},x) =
\bigl\langle \mu(x,\cdot),\, G(x)\,Q(x)^D\,e_k
\bigr\rangle_{L^2(\mathcal{A})},
\]
where $e_k$ is the $k$-th standard basis vector in $\mathbb{R}^{n_{\cal Y}}$,
from which the result follows.

\subsubsection{Proof of Theorem~\ref{theorem.inf}}

By \eqref{eq:Winf}, we have
\[
\mathbb{E}\bigl[w^{(\infty)}(Y,X)\,| \, X=x\bigr]  
=  \mu(x,\cdot)^*\,G(x)\,Q(x)^D\,
\mathbf{F}(x)\,\pi_0(\cdot\,| \, x).
\]
As a result, 
\begin{equation}\label{eq:bias_decomp1}
\begin{aligned}
&\mathbb{E}\bigl[w^{(\infty)}(Y,X)\,| \, X=x\bigr] 
-\mathbb{E}\bigl[\mu(X,A)\,| \, X=x\bigr]\\
&\quad=\mu(x,\cdot)^*\bigl\{G(x)\,Q(x)^D\,
\mathbf{F}(x)-\mathbb{I}_{n_{\cal Y}}\bigr\}\pi_0(\cdot\,| \, x)\\
&\quad=\mu(x,\cdot)^*\bigl\{G(x)\,Q(x)^D\,
\mathbf{F}(x)-\mathbb{I}_{n_{\cal Y}}\bigr\}
\pi_{\mathcal{P}(x)^\perp}(\cdot\,| \, x)\\
&\qquad+\mu(x,\cdot)^*\bigl\{G(x)\,Q(x)^D\,
\mathbf{F}(x)-\mathbb{I}_{n_{\cal Y}}\bigr\}
\pi_{\mathcal{P}(x)}(\cdot\,| \, x)\\
&\quad=\mu(x,\cdot)^*\bigl\{G(x)\,Q(x)^D\,
\mathbf{F}(x)-\mathbb{I}_{n_{\cal Y}}\bigr\}
\pi_{\mathcal{P}(x)^\perp}(\cdot\,| \, x),
\end{aligned}
\end{equation}
where in the last equality we used Lemma~\ref{lemma:equality} and the fact that, 
by definition of $\pi_{\mathcal{P}(x)}(\cdot\,| \, x)$, 
$\pi_{\mathcal{P}(x)}(\cdot\,| \, x)=\pi_{\rm p}(x)\,\mathbf{F}(x)^*\,
\nu(\cdot,x)$ for some $\nu(\cdot\,| \, x):\mathcal{Y}\to\mathbb{R}$.
Then write 
$\mu(x,\cdot)=\mu_{\mathcal{F}(x)}(x,\cdot)+ 
\mu_{\mathcal{F}(x)^\perp}(x,\cdot)$.
Since the prior is uniform, $\pi_{\rm p}(x)$ is a scalar multiple of the identity, and 
therefore 
\[
\operatorname{range}\bigl\{\pi_{\rm p}(x)\,\mathbf{F}(x)^*\bigr\}
=\operatorname{range}\bigl\{\mathbf{F}(x)^*\bigr\}.
\]
Hence,
\[
\mu_{\mathcal{F}(x)^\perp}(x,\cdot)^*\,G(x)= 
\mu_{\mathcal{F}(x)^\perp}(x,\cdot)^*\,\pi_{\rm p}(x)\,
\mathbf{F}(x)^*\,P_{\rm p}(x)^{-1}=0.
\]
By \eqref{eq:bias_decomp1}, this leads to 
\begin{equation}\label{eq:bias_decomp2}
\begin{aligned}
&\mathbb{E}\bigl[w^{(\infty)}(Y,X)\,| \, X=x\bigr] 
-\mathbb{E}\bigl[\mu(X,A)\,| \, X=x\bigr]\\
&\quad=\mu_{\mathcal{F}(x)}(x,\cdot)^*\bigl\{G(x)\,
Q(x)^D\,\mathbf{F}(x)-\mathbb{I}_{n_{\cal Y}}\bigr\}
\pi_{\mathcal{P}(x)^\perp}(\cdot\,| \, x)\\
&\qquad-\mu_{\mathcal{F}(x)^\perp}(x,\cdot)^*\,
\pi_{\mathcal{P}(x)^\perp}(\cdot\,| \, x).
\end{aligned}
\end{equation}
Write $\mu_{\mathcal{F}(x)}=\mathbf{F}(x)^*\widetilde{\mu}_{\mathcal{F}(x)}$ 
for some $\widetilde{\mu}_{\mathcal{F}(x)}\in\mathbb{R}^{n_{\cal Y}}$. Then
\begin{equation}\label{eq:bias_decomp3}
\begin{aligned}
&\mu_{\mathcal{F}(x)}(x,\cdot)^*\bigl\{G(x)\,
Q(x)^D\,\mathbf{F}(x)-\mathbb{I}_{n_{\cal Y}}\bigr\}\\
&\quad=\widetilde{\mu}_{\mathcal{F}(x)}^\top\,\mathbf{F}(x)\,
\bigl\{G(x)\,Q(x)^D\,\mathbf{F}(x)
-\mathbb{I}_{n_{\cal Y}}\bigr\}\\
&\quad=\widetilde{\mu}_{\mathcal{F}(x)}^\top\,\mathbf{F}(x)\,G(x)\,
Q(x)^D\,\mathbf{F}(x)
-\widetilde{\mu}_{\mathcal{F}(x)}^\top\,\mathbf{F}(x).
\end{aligned}
\end{equation}
Next,
\begin{equation}\label{eq:bias_decomp4}
\begin{aligned}
&\widetilde{\mu}_{\mathcal{F}(x)}^\top\,\mathbf{F}(x)\,G(x)\,
Q(x)^D\,\mathbf{F}(x)\\
&=\widetilde{\mu}_{\mathcal{F}(x)}^\top\,Q(x)\,Q(x)^D\,
\mathbf{F}(x)\\
&=\widetilde{\mu}_{\mathcal{F}(x)}^\top\,U(x)\,
\diag\bigl[\mathbf{1}\{\lambda_k(x)\ne 0\}\bigr]_{k=1,\ldots,n_{\cal Y}}\,
U^{-1}(x)\,\mathbf{F}(x).
\end{aligned}
\end{equation}
Now notice that 
\begin{align*}
&\diag\bigl[\mathbf{1}\{\lambda_k(x)=0\}\bigr]_{k=1,\ldots,n_{\cal Y}}\,
U^{-1}(x)\,Q(x)\\
&=\diag\bigl[\lambda_k(x)\,
\mathbf{1}\{\lambda_k(x)=0\}\bigr]_{k=1,\ldots,n_{\cal Y}}\,
U^{-1}(x)=0.
\end{align*}
By Lemma~3 in \cite{dhaene2023approximate}, this yields 
\[
\diag\bigl[\mathbf{1}\{\lambda_k(x)=0\}\bigr]_{k=1,\ldots,n_{\cal Y}}\,
U^{-1}(x)\,\mathbf{F}(x)=0,
\]
and therefore
\begin{align*}
&U(x)\,\diag\bigl[\mathbf{1}\{\lambda_k(x)\ne 0\}
\bigr]_{k=1,\ldots,n_{\cal Y}}\,
U^{-1}(x)\,\mathbf{F}(x)\\
&\quad=U(x)\,U^{-1}(x)\,\mathbf{F}(x)
=\mathbf{F}(x).
\end{align*}
Substituting back into \eqref{eq:bias_decomp4} and then \eqref{eq:bias_decomp3} gives
\[
\mu_{\mathcal{F}(x)}(x,\cdot)^*\bigl\{G(x)\,
Q(x)^D\,\mathbf{F}(x)-\mathbb{I}_{n_{\cal Y}}\bigr\}
\pi_{\mathcal{P}(x)^\perp}(\cdot\,| \, x)=0.
\]
By \eqref{eq:bias_decomp2}, we therefore obtain
\[
\mathbb{E}\bigl[w^{(\infty)}(Y,X)\,| \, X=x\bigr] 
-\mathbb{E}\bigl[\mu(X,A)\,| \, X=x\bigr]
=-\mu_{\mathcal{F}(x)^\perp}(x,\cdot)^*\,
\pi_{\mathcal{P}(x)^\perp}(\cdot\,| \, x),
\]
and the result of the theorem follows by taking expectations over $X$.
\subsubsection{Auxiliary lemmas}

\begin{lemma}\label{lemma.pseudoinverse}
Let $H$ be a (real) Hilbert space and let $B: H \to \mathbb{R}^m$ be a bounded 
linear operator. Set $A := B B^{*}$, where $B^{*}: \mathbb{R}^m \to H$. Then
\[
B^{*} A^{\dagger} A \;=\; B^{*},
\qquad\text{equivalently,}\qquad
(A^{\dagger}A)B=B.
\]
\end{lemma}

\begin{proof}
Since the codomain of $B$ is finite-dimensional, $B$ admits a Schmidt decomposition 
\citep[see, e.g.,][]{horn1994topics}: there exist singular values 
$\sigma_1\ge \cdots \ge \sigma_r>0$, an orthonormal family 
$\{u_i\}_{i=1}^r \subset \mathbb{R}^m$ spanning $\operatorname{range}(B)$, and an 
orthonormal family $\{v_i\}_{i=1}^r \subset H$ such that, for all $x\in H$ and 
$y\in \mathbb{R}^m$,
\[
Bx = \sum_{i=1}^r \sigma_i \,\langle x, v_i\rangle_H\, u_i,
\qquad
B^{*}y = \sum_{i=1}^r \sigma_i \,\langle y, u_i\rangle_{\mathbb{R}^m}\, v_i.
\]
Consequently,
\[
A \;=\; B B^{*} \;=\; \sum_{i=1}^r \sigma_i^2\, u_i u_i^\top
\quad\text{on }\mathbb{R}^m,
\]
so $A u_i = \sigma_i^2 u_i$ for $i=1,\dots,r$, and $A$ vanishes on 
$\operatorname{ran}(B)^{\perp}$. Hence the pseudoinverse is
\[
A^{\dagger} \;=\; \sum_{i=1}^r \frac{1}{\sigma_i^2}\, u_i u_i^\top,
\]
and thus
\[
A^{\dagger}A \;=\; \sum_{i=1}^r u_i u_i^\top \;=\; P_{\operatorname{range}(B)},
\]
the orthogonal projector in $\mathbb{R}^m$ onto $\operatorname{range}(B)$.
For any $x\in H$, we have $Bx \in \operatorname{range}(B)$, whence
\[
(A^{\dagger}A) Bx \;=\; P_{\operatorname{range}(B)} (Bx) \;=\; Bx.
\]
Therefore $(A^{\dagger}A)B = B$. Taking Hilbert adjoints yields 
$B^{*}A^{\dagger}A = B^{*}$, as claimed.
\end{proof}

\begin{lemma}\label{lemma:equality} 
It holds that 
$$G(x) Q(x)^D \mathbf{F}(x)
\pi_{\rm p}(x)\mathbf{F}(x)^*=\pi_{\rm p}(x)\mathbf{F}(x)^*.$$
\end{lemma}

\begin{proof}

Applying Lemma~\ref{lemma.pseudoinverse} with $B=\widetilde{\mathbf{F}}(x)$ 
and $A=\widetilde{Q}(x)$ and using that, since $\widetilde{Q}(x)$ is symmetric, we have $\widetilde{Q}(x)^{\dagger}=\widetilde{Q}(x)^{D}$, we obtain
\begin{equation}\label{24042026}
\widetilde{\mathbf{F}}(x)^* \, \widetilde{Q}(x)^D \,    
\widetilde{\mathbf{F}}(x) \widetilde{\mathbf{F}}(x)^* 
= \widetilde{\mathbf{F}}(x)^*.
\end{equation}
By definition of $\widetilde{\mathbf{F}}(x)$ and 
$\widetilde{Q}(x)$, we also have 
\begin{equation}\label{240420261}
\begin{aligned}
&\pi_{\rm p}(x)^{1/2}\widetilde{\mathbf{F}}(x)^* \, \widetilde{Q}(x)^D \,    
\widetilde{\mathbf{F}}(x) \widetilde{\mathbf{F}}(x)^* P_{\rm p}(x)^{1/2}\\
&=\pi_{\rm p}(x) \mathbf{F}(x)^* P_{\rm p}(x)^{-1/2}  P_{\rm p}(x)^{-1/2} 
Q(x)^D  P_{\rm p}(x)^{-1/2} P_{\rm p}(x)^{1/2} \\
&\quad \times 
\mathbf{F}(x)\pi_{\rm p}(x)^{1/2}\pi_{\rm p}(x)^{1/2} \mathbf{F}(x)^*
P_{\rm p}(x)^{-1/2} P_{\rm p}(x)^{1/2}\\
&= \pi_{\rm p}(x) \mathbf{F}(x)^* P_{\rm p}(x)^{-1}
Q(x)^D
\mathbf{F}(x)\pi_{\rm p}(x) \mathbf{F}(x)^*\\
&= G(x) Q(x)^D \mathbf{F}(x)
\pi_{\rm p}(x)\mathbf{F}(x)^*,
\end{aligned}
\end{equation}
and 
\begin{equation}
    \label{240420262}
    \begin{aligned}
    &\pi_{\rm p}(x)^{1/2}\widetilde{\mathbf{F}}(x)^* P_{\rm p}(x)^{1/2}\\
    &=\pi_{\rm p}(x)^{1/2}\pi_{\rm p}(x)^{1/2}\mathbf{F}(x)^* P_{\rm p}(x)^{-1/2}P_{\rm p}(x)^{1/2}\\
    &=    \pi_{\rm p}(x) \mathbf{F}(x)^* .
    \end{aligned}
\end{equation}
We obtain the result by combining \eqref{24042026}, \eqref{240420261} and \eqref{240420262}.
\end{proof}

\subsection{Proofs of the results of Section~\ref{sec:InferenceKnownTheta}}

\subsubsection{Proof of Theorem~\ref{thm:asymptotic_normality}}

Decompose the estimation error as
\begin{equation}\label{eq:decomposition}
\frac{\sqrt{n}\left(\widehat{\mu}^{(\infty)} - \mu_0\right)}{\sigma_T} =
\frac{\sqrt{n}\left(\widehat{\mu}^{(\infty)} - \mu_*^{(\infty)}\right)}{\sigma_T} +
\frac{\sqrt{n}\left(\mu_*^{(\infty)} - \mu_0\right)}{\sigma_T}.
\end{equation}

\noindent
\textbf{Step 1 (Asymptotic normality of the sampling error).}
Define, for each $n$ and $i = 1, \ldots, n$,
$$
Z_{n,i} := \frac{w^{(\infty)}(Y_i,X_i) - \mu_*^{(\infty)}}{\sigma_T}.
$$
Under Assumption~\ref{ass:iid}, for each $n$, the random variables 
$Z_{n,1}, \ldots, Z_{n,n}$ are i.i.d.\ with $\mathbb{E}[Z_{n,i}] = 0$ and 
${\rm Var}(Z_{n,i}) = 1$. We have
$$
\frac{\sqrt{n}\left(\widehat{\mu}^{(\infty)} - \mu_*^{(\infty)}\right)}{\sigma_T}
= \frac{1}{\sqrt{n}} \sum_{i=1}^n Z_{n,i}.
$$
This is a triangular array of rowwise i.i.d.\ random variables with zero mean and 
unit variance. By the Lindeberg--Feller central limit theorem 
\citep{billingsley2013convergence}, convergence to ${\cal N}(0,1)$ holds if the Lindeberg 
condition is satisfied: for all $\epsilon > 0$,
$$
\frac{1}{n} \sum_{i=1}^n \mathbb{E}\left[Z_{n,i}^2 \, 
\mathbf{1}\{|Z_{n,i}| > \epsilon\sqrt{n}\}\right] \to 0.
$$
Since the $Z_{n,i}$ are identically distributed for each $n$, this reduces to
$$
\mathbb{E}\left[Z_{n,1}^2 \, \mathbf{1}\{|Z_{n,1}| > \epsilon\sqrt{n}\}\right] \to 0,
$$
which is precisely Assumption~\ref{ass:lindeberg}. Therefore,
\begin{equation}\label{eq:clt}
\frac{\sqrt{n}\left(\widehat{\mu}^{(\infty)} - \mu_*^{(\infty)}\right)}{\sigma_T}
\xrightarrow{d} {\cal N}(0,1).
\end{equation}

\noindent
\textbf{Step 2 (Asymptotic negligibility of the bias).}
We have
$$
\left|\frac{\sqrt{n}\left(\mu_*^{(\infty)} - \mu_0\right)}{\sigma_T}\right| 
\leq \frac{\sqrt{n} \cdot O(r_T)}{\underline{\sigma}} = o(1)
$$
under condition~\eqref{eq:rate_condition} and 
Assumption~\ref{ass:variance_lower_bound}. Hence,
\begin{equation}\label{eq:bias_negligible}
\frac{\sqrt{n}\left(\mu_*^{(\infty)} - \mu_0\right)}{\sigma_T} \xrightarrow{p} 0.
\end{equation}

\noindent
\textbf{Step 3.}
Since the first term in \eqref{eq:decomposition} converges in distribution to 
${\cal N}(0,1)$ by \eqref{eq:clt} and the second converges in probability to zero by 
\eqref{eq:bias_negligible}, Slutsky's theorem yields
$$
\frac{\sqrt{n}\left(\widehat{\mu}^{(\infty)} - \mu_0\right)}{\sigma_T}
\xrightarrow{d} {\cal N}(0,1).
$$

\subsubsection{Proof of Corollary~\ref{cor:feasible}}

We show that $\widehat{\sigma}_T / \sigma_T \xrightarrow{p} 1$, which combined with 
Theorem~\ref{thm:asymptotic_normality} and Slutsky's theorem yields the result.
Write
$$
\frac{\widehat{\sigma}_T^2}{\sigma_T^2} = \frac{1}{n\sigma_T^2}\sum_{i=1}^n W_i^2 
- \frac{\left(\widehat{\mu}^{(\infty)} - \mu_*^{(\infty)}\right)^2}{\sigma_T^2},
$$
where $W_i := w^{(\infty)}(Y_i,X_i) - \mu_*^{(\infty)}$.

\noindent
\textbf{Step 1.} We show 
$\frac{1}{n\sigma_T^2}\sum_{i=1}^n W_i^2 \xrightarrow{p} 1$.
Define $V_i := W_i^2 / \sigma_T^2$. Then $\mathbb{E}[V_i] = 1$ and 
$\mathbb{E}[V_i^2] = \mathbb{E}[W_i^4]/\sigma_T^4 \leq \kappa$ 
by Assumption~\ref{ass:fourth_moment}. By Chebyshev's inequality,
$$
P\left(\left|\frac{1}{n}\sum_{i=1}^n V_i - 1\right| > \epsilon\right) 
\leq \frac{{\rm Var}(V_1)}{n\epsilon^2} 
\leq \frac{\mathbb{E}[V_1^2]}{n\epsilon^2} 
\leq \frac{\kappa}{n\epsilon^2} \to 0.
$$

\noindent
\textbf{Step 2.} We show 
$\left(\widehat{\mu}^{(\infty)} - \mu_*^{(\infty)}\right)^2 / \sigma_T^2
\xrightarrow{p} 0$.
We have
$\mathbb{E}\left[\left(\widehat{\mu}^{(\infty)} - \mu_*^{(\infty)}\right)^2\right]
= {\rm Var}(\widehat{\mu}^{(\infty)}) = \sigma_T^2/n$. Thus,
$$
\mathbb{E}\left[\frac{\left(\widehat{\mu}^{(\infty)} - 
\mu_*^{(\infty)}\right)^2}{\sigma_T^2}\right] 
= \frac{1}{n} \to 0,
$$
and the result follows by Markov's inequality.

\noindent
\textbf{Step 3.} Combining Steps 1 and 2, 
$\widehat{\sigma}_T^2 / \sigma_T^2 \xrightarrow{p} 1$. 
Since $\sigma_T^2 \geq \underline{\sigma}^2 > 0$ by 
Assumption~\ref{ass:variance_lower_bound}, the square root function is continuous 
at the limit point, and the continuous mapping theorem gives
$\widehat{\sigma}_T / \sigma_T \xrightarrow{p} 1$.
The result follows from Theorem~\ref{thm:asymptotic_normality} and Slutsky's theorem.

\subsection{Proof of the bias bound in Section~\ref{subsec:regimeC-example}}
\label{app:proof-lemma1}

\begin{proof}
Let $S:=T/2$. Conditional on $A=(\alpha_1,\alpha_2)^\top$, the success probabilities are
\begin{equation}
\label{p1p2}
  p_1 := F(\alpha_1), \qquad p_2 := F\left(\alpha_1 + \frac{\alpha_2}{\sqrt{T}}\right),
\end{equation}
where $F(z) = 1/(1+e^{-z})$.
The sufficient statistics
\[
  Y_1:= \sum_{t=1}^{S} Y_{t}, \qquad Y_2 := \sum_{t=S+1}^{T} Y_{t}
\]
are independent with $Y_1 \sim \mathrm{Bin}(S, p_1)$ and 
$Y_2 \sim \mathrm{Bin}(S, p_2)$.

We work in the $(p_1,p_2)$ parameterization. From \eqref{p1p2}, we have
\[
  \alpha_1=\logit(p_1), \qquad   \alpha_2=\sqrt{T}\bigl(\logit(p_2)-\logit(p_1)\bigr),
\]
and therefore the target $F(\alpha_1+\alpha_2)$ can be written as
\begin{equation}\label{eq:gT}
  g_T(p_1, p_2) := F\Bigl((1-\sqrt{T})\,\logit(p_1) + \sqrt{T}\,\logit(p_2)\Bigr).
\end{equation}

\paragraph{Factorial moment identity.}
For $Y \sim \mathrm{Bin}(n, p)$ and any integer $0 \leq r \leq n$, define
the falling factorial ratio
\begin{equation}\label{eq:factorial}
  \frac{(Y)_r}{(n)_r} := 
  \frac{Y(Y-1)\cdots(Y-r+1)}{n(n-1)\cdots(n-r+1)},
\end{equation}
with the convention $(Y)_0 = (n)_0 = 1$.
Then $\mathbb{E}\!\left[\frac{(Y)_r}{(n)_r} \right] = p^r$, so the
falling factorial ratio is an unbiased estimator of $p^r$.

\paragraph{Chebyshev interpolation on $[\epsilon,1-\epsilon]$.}
Fix $\epsilon \in (0,1/2)$.
Define the Chebyshev nodes on $[\epsilon, 1-\epsilon]$:
\begin{equation*}\label{eq:cheb_nodes}
  t_m = \frac{1}{2} + \frac{1 - 2\epsilon}{2}\,
  \cos\!\left(\frac{(2m+1)\pi}{2(S+1)}\right),
  \qquad m = 0, 1, \ldots, S,
\end{equation*}
and the corresponding Lagrange basis polynomials
\begin{equation*}\label{eq:lagrange}
  \ell_m(p) = \prod_{\substack{r=0 \\ r \neq m}}^{S} \frac{p - t_r}{t_m - t_r}, 
  \qquad m = 0, \ldots, S.
\end{equation*}
Expand in monomials: $\ell_m(p) = \sum_{s=0}^{S} c_{ms}\, p^s$.

\paragraph{Unbiased implementation.}
For $y \in \{0, 1, \ldots, S\}$, define
\begin{equation*}\label{eq:L}
  L_m(y) := \sum_{s=0}^{S} c_{ms}\, \frac{(y)_s}{(S)_s}.
\end{equation*}
By the factorial moment identity, 
$\mathbb{E}[L_m(Y) ] = \ell_m(p)$ for $Y \sim \mathrm{Bin}(S, p)$.

\paragraph{Estimator.} Define
\begin{equation*}\label{eq:mT}
  m_T(Y_1, Y_2) = \sum_{m=0}^{S}\sum_{l=0}^{S}
    g_T(t_m, t_l)\; L_m(Y_1)\, L_l(Y_2),
\end{equation*}
with $g_T$ as defined in \eqref{eq:gT}.
Taking expectations and using the independence of $Y_1$ and $Y_2$ yields
\begin{equation*}\label{eq:interp}
  \mathbb{E}[m_T(Y_1, Y_2) ]
  = \sum_{m=0}^{S}\sum_{l=0}^{S} g_T(t_m, t_l) \, \ell_m(p_1)\, \ell_l(p_2)
  =: I_S[g_T](p_1, p_2).
\end{equation*}
The bias equals the interpolation error:
\begin{equation*}\label{eq:bias_interp}
  \mathbb{E}[m_T(Y_1,Y_2) ] - g_T(p_1, p_2)
  = I_S[g_T](p_1, p_2) - g_T(p_1, p_2).
\end{equation*}
It remains to bound the right-hand side.

\paragraph{Bounding the interpolation error.}
Let $\epsilon\in(0,1/2)$ be such that $p_1,p_2\in[\epsilon,1-\epsilon]$.
For fixed $p_1 \in [\epsilon, 1-\epsilon]$, the function $p_2 \mapsto g_T(p_1, p_2)$
is analytic on $(0,1)$ and extends to a meromorphic function of $p_2 \in \mathbb{C}$.
The logistic function $F(z) = 1/(1+e^{-z})$ has poles where $1 + e^{-z} = 0$,
i.e., at $z = (2k+1)\pi i$ for $k \in \mathbb{Z}$.

Fix $p_1 \in [\epsilon, 1-\epsilon]$ real. The singularities of 
$p_2\in\mathbb{C} \mapsto g_T(p_1, p_2)$ occur when
\[
  \underbrace{(1-\sqrt{T})\,\logit(p_1)}_{\in\,\mathbb{R}}
  \;+\; \sqrt{T}\,\logit(p_2)
  = (2k+1)\pi i,
\]
which requires $\mathrm{Im}\bigl[\sqrt{T}\,\logit(p_2)\bigr] = (2k+1)\pi$.
The nearest singularity satisfies $|\mathrm{Im}[\logit(p_2)]| = \pi/\sqrt{T}$.

Since $|\logit'(p)| = 1/(p(1-p)) \leq 1/(\epsilon(1-\epsilon))$
on $[\epsilon,1-\epsilon]$, a mean-value argument shows that the imaginary displacement 
$|\mathrm{Im}(\logit(p_2))| = \pi/\sqrt{T}$ requires a displacement from the real 
interval of at least
\[
  \delta_T \geq \frac{\pi\, \epsilon(1-\epsilon)}{\sqrt{T}}.
\]
That is, the nearest singularity to $[\epsilon,1-\epsilon]$ in the complex $p_2$-plane 
lies at imaginary distance at least $\delta_T$.

Set $\kappa := 2\pi\epsilon(1-\epsilon)/(1-2\epsilon) > 0$
and $\rho_T := 1 + \kappa/\sqrt{T}$. The corresponding Bernstein ellipse
$\mathcal{E}_{\rho_T}$ around $[\epsilon, 1-\epsilon]$ has semi-minor axis
$(1-2\epsilon)(\rho_T-1/\rho_T)/4 \le \delta_T$ (a direct calculation shows
this inequality holds for every $T \ge 1$), so $\mathcal{E}_{\rho_T}$
is contained in the analyticity region of $p_2 \mapsto g_T(p_1,p_2)$
derived above.
By classical results on polynomial interpolation at Chebyshev nodes
\citep[Theorem~8.2]{trefethen2019approximation},
the one-dimensional interpolation error satisfies
\begin{equation}\label{eq:1d-interp}
  \|I_S[g_T(p_1,\cdot)] - g_T(p_1,\cdot)\|_{L^\infty([\epsilon,1-\epsilon])}
  \leq \frac{4M}{\rho_T - 1}\, \rho_T^{-S}
  = \frac{4M\sqrt{T}}{\kappa}\, \rho_T^{-S},
\end{equation}
where $M = \sup_{\mathcal{E}_{\rho_T}} |g_T(p_1,\cdot)| < \infty$ uniformly in
$p_1$ and $T$, since the Bernstein ellipse stays bounded away from the poles of $F$.
The factor $1/(\rho_T-1)=\sqrt{T}/\kappa$ contributes a $\sqrt{T}$ pre-factor
that must be tracked explicitly.

To pass to the two-dimensional case, recall that $I_S[g_T]$ is the tensor-product
Chebyshev interpolant. Let $I_S^{(1)}$ and $I_S^{(2)}$ denote the one-dimensional
interpolation operators in $p_1$ and $p_2$ respectively. Then
\[
g_T - I_S[g_T]
= (g_T - I_S^{(1)} g_T) + I_S^{(1)}(g_T - I_S^{(2)} g_T).
\]
The operator norm of $I_S^{(1)}$ on $L^\infty([\epsilon,1-\epsilon])$ equals the
Lebesgue constant $\Lambda_S$ of Chebyshev interpolation, which satisfies
$\Lambda_S \leq \tfrac{2}{\pi}\log(S+1) + 1$
\citep[Theorem~15.2]{trefethen2019approximation}.
Combining the decomposition above with \eqref{eq:1d-interp} applied in each
coordinate, we obtain
\begin{equation}\label{eq:2d-interp}
\|I_S[g_T] - g_T\|_{L^\infty([\epsilon,1-\epsilon]^2)}
\;\leq\; (1 + \Lambda_S)\,\frac{4M\sqrt{T}}{\kappa}\, \rho_T^{-S}
\;\leq\; C_1\,\sqrt{T}\,\log(T+1)\, \rho_T^{-S},
\end{equation}
for a constant $C_1$ independent of $T$.

It remains to bound $\rho_T^{-S}$. Since $S=T/2$ and $\rho_T = 1 + \kappa/\sqrt{T}$,
the inequality $\log(1+x) \geq x - x^2/2$ for $x \geq 0$ yields
\[
\rho_T^{-S}
= \exp\!\Bigl(-\tfrac{T}{2}\log\rho_T\Bigr)
\leq \exp\!\Bigl(-\tfrac{T}{2}\bigl(\tfrac{\kappa}{\sqrt{T}} - \tfrac{\kappa^2}{2T}\bigr)\Bigr)
= e^{\kappa^2/4}\,e^{-\kappa\sqrt{T}/2}.
\]
Substituting into \eqref{eq:2d-interp},
\[
\|I_S[g_T] - g_T\|_{L^\infty([\epsilon,1-\epsilon]^2)}
\;\leq\; C_2\,\sqrt{T}\,\log(T+1)\,e^{-\kappa\sqrt{T}/2}.
\]
Finally, fix any $c_0 \in (0, \kappa/2)$. Since
$\sqrt{T}\,\log(T+1)\, e^{-(\kappa/2 - c_0)\sqrt{T}} \to 0$ as $T \to \infty$,
this quantity is bounded by some constant $C_3 = C_3(c_0)$ uniformly in $T$,
and hence
\[
\|I_S[g_T] - g_T\|_{L^\infty([\epsilon,1-\epsilon]^2)}
\;\leq\; C\, e^{-c_0\sqrt{T}}
\]
for $C = C_2 C_3$. The polynomial $\sqrt{T}\log(T+1)$ pre-factor has been
absorbed into the exponential at the cost of a strictly smaller exponent
constant. This establishes the claimed bound.
\end{proof}

\end{document}